\documentclass{article}
\usepackage[utf8]{inputenc}
\usepackage{amssymb}
\usepackage{amsmath}
\usepackage{wasysym}
\usepackage{graphicx}
\usepackage{wrapfig}

\newtheorem{definition}{Definition}
\newtheorem{theorem}{Theorem}
\newtheorem{lemma}{Lemma}
\newtheorem{example}{Example}
\newtheorem{corollary}{Corollary}

\newtheorem{claim}{\em Claim}
\newenvironment{proof}{\noindent{\sf Proof.}}{\hfill $\boxtimes$\linebreak}
\newcommand{\qed}{\hfill $\boxtimes$\linebreak}
\renewcommand{\phi}{\varphi}
\renewcommand{\Sigma}{Sig}
\newcommand{\miniminus}{\mbox{-}}

\title{Epistemic Logic of Networks}
\title{Epistemic Logic for Networks}
\title{Epistemic Logic for Communication Networks}
\title{Knowledge over Communication Networks}
\title{Knowledge in Communication Networks}

\author{Pavel G. Naumov$^\star$ \and  Jia Tao$^\ast$}

\begin{document}

\maketitle

{\let\thefootnote\relax\footnotetext{$^\star$ Department of Computer  Science, Illinois Wesleyan University, Bloomington, Illinois, the United States, {\sf pavel@pavelnaumov.com}}

\let\thefootnote\relax\footnotetext{ $^\ast$ Department of Computer Science, The College of New Jersey, Ewing, New Jersey, the United States, {\sf taoj@tcnj.edu}}}

\begin{abstract}
The article investigates epistemic properties of information flow under communication protocols with a given topological structure of the communication network. The main result is a sound and complete logical system that describes all such properties. The system consists of a variation of the multi-agent epistemic logic S5 extended by a new network-specific Gateway axiom. 
\end{abstract}

\section{Introduction}\label{introduction}


In this article we study epistemic properties of communication protocols. Consider, for example, a protocol $\mathcal{P}_1$ between agents $p$, $q$, $u$, and $v$.
Under this protocol, agent $p$ communicates to agent $q$ a message over a secure communication channel $m$. Next, agent $q$ must communicate the same message over insecure channels to agent $u$. To achieve this, agent $q$ chooses a random one-time encryption pad (``key") and computes a ciphertext as a bit-wise sum of the message and the key modulo 2. Agent $q$ then sends the key and the ciphertext to agent $u$ over insecure channels $k$ and $c$ accordingly. Finally, agent $u$, upon receiving the key and the ciphertext, computes a bit-wise sum of these two strings modulo 2 and communicates the result over a secure channel $m'$ to agent $v$. 

\begin{figure}[ht]
\begin{center}
\vspace{3mm}
\scalebox{.6}{\includegraphics{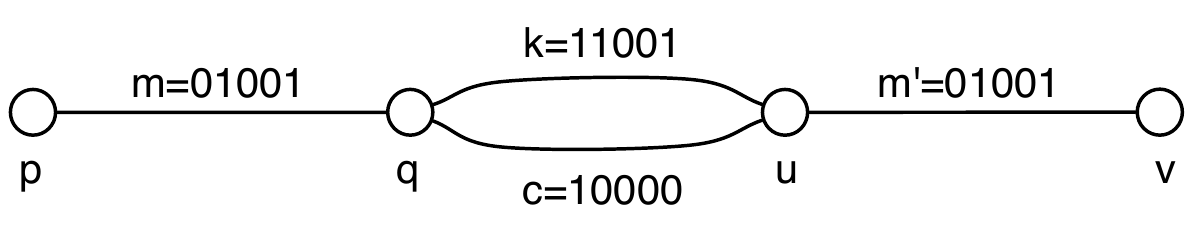}}
\vspace{0mm}
\caption{Run $r_1$ of protocol ${\cal P}_1$.}\label{example1 figure}
\vspace{0cm}
\end{center}
\vspace{0cm}
\end{figure}

A {\em run} of a protocol is an assignment of values to all communication channels that satisfy the restrictions imposed by the protocol. An example of a run $r_1$ of protocol ${\cal P}_1$ is depicted in Figure~\ref{example1 figure}. Note that for any run satisfying the restrictions of ${\cal P}_1$, the value of channel $m$ is the same as the value of channel $m'$. Thus, any outside observer who can eavesdrop on channel $m$ under run $r_1$ would be able to learn that channel $m'$ has a value of $01001$ on this run. Using epistemic modal logic notations\footnote{Similarly to Kane and Naumov~\cite{kn13tark}, we interpret modality $\Box_m$ as ``any outside observer who can eavesdrop on channel $m$ knows that \dots", instead of more traditional ``agent $m$ knows that \dots"~\cite{fhmv95}.}, we write this  as
$$r_1\Vdash \Box_{m}(m'=01001).$$

At the same time, since there is no connection between the values of the ciphertext $c$ and the original message $m$, an external observer eavesdropping on channel $c$ would not be able to deduce the value of channel $m'$:
\begin{equation}\label{c not knows}
r_1\Vdash \neg\Box_{c}(m'=01001).
\end{equation}
Similarly,
\begin{equation}\label{k not knows}
r_1\Vdash \neg\Box_{k}(m'=01001).
\end{equation}

We now consider a variation of protocol ${\cal P}_1$ that we call $\mathcal{P}_2$. Under the second protocol agents $q$ and $u$ are allowed to make a mistake in at most one bit during the encryption and the decryption stages respectively. In other words, the Hamming distance between the value of channel $c$ and the bit-wise sum of values of channels $m$ and $k$ is no more than one. Similarly, the Hamming distance between the value of channel $m'$ and the bit-wise sum of values of channels $c$ and $k$ is no more than one. An example of a run $r_2$ of protocol $\mathcal{P}_2$ is depicted on Figure~\ref{example2 figure}. 

\begin{figure}[ht]
\begin{center}
\vspace{3mm}
\scalebox{.6}{\includegraphics{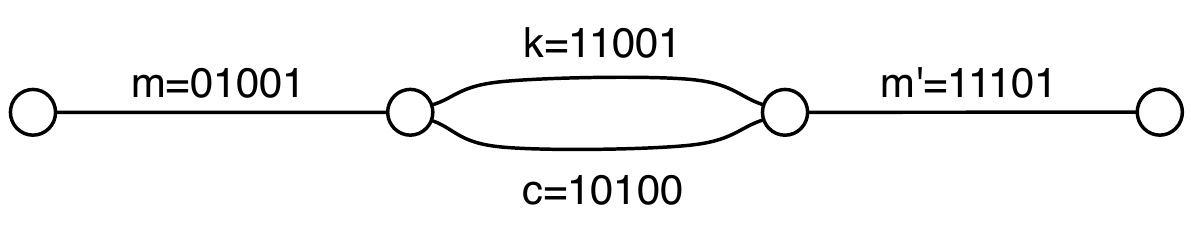}}
\vspace{0mm}
\caption{Run $r_2$ of protocol ${\cal P}_2$.}\label{example2 figure}
\vspace{0cm}
\end{center}
\vspace{0cm}
\end{figure}

Note that run $r_1$ is also a valid run of protocol $\mathcal{P}_2$. Thus, an external observer eavesdropping on channel $m$ on run $r_2$ is not able to distinguish run $r_2$ from run $r_1$. Hence, such an observer would not be able to conclude that the value of $m'$ is $11101$. Therefore, under protocol $\mathcal{P}_2$,
$$
r_2\Vdash\neg\Box_m(m'=11101).
$$
At the same time, an external observer eavesdropping on channel $m$ on run $r_2$ of protocol $\mathcal{P}_2$ should be able to conclude that the value of channel $m'$ is {\em not} $01110$ because the Hamming distance between $01001$ and $01110$ is three and, according to the restrictions of protocol $\mathcal{P}_2$, errors could be introduced in at most two bits during the encryption and the decryption stages combined:
$$
r_2\Vdash\Box_m(m'\neq 01110).
$$

We now consider another variation of protocol $\mathcal{P}_1$ that we call ${\cal P}_3$, see Figure~\ref{example3 figure}. The original message $m$ in this protocol is first encrypted into a cyphertext $c$ using a key $k$, then it is recovered as $m'$, then again encrypted and recovered as $m''$. A single bit-error could be introduced by each encryption and decryption stage. Thus, the Hamming distance between strings $m$ and $m''$ could be at most four. Figure~\ref{example3 figure} shows a possible run $r_3$ of this protocol.

\begin{figure}[ht]
\begin{center}
\vspace{3mm}
\scalebox{.6}{\includegraphics{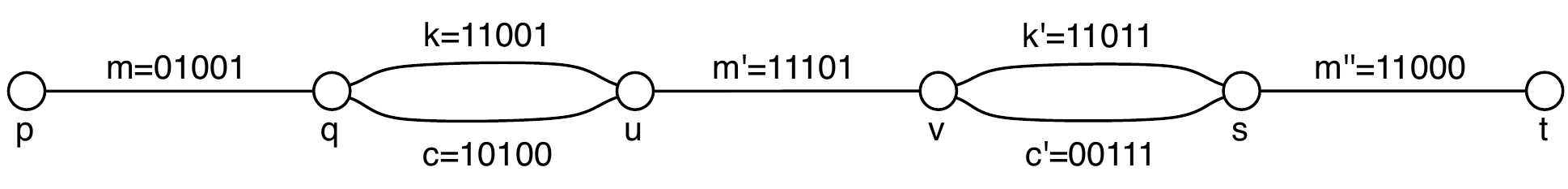}}
\vspace{0mm}
\caption{Run $r_3$ of protocol ${\cal P}_3$.}\label{example3 figure}
\vspace{0cm}
\end{center}
\vspace{0cm}
\end{figure}

An external observer eavesdropping on channel $m$ on run $r_3$ under $\mathcal{P}_3$ would not be able to know the exact value of $m'$. However, it would know that the value of channel $m'$ is at a Hamming distance no more than two from the value of $m$. Note that the Hamming distance between the value of $m$ and the string $10110$ is five. Thus, due to the triangle inequality, the observer would be able to conclude that the Hamming distance between the value of $m'$ and the string $10110$ is at least three. Based on this, the observer would be able to conclude that any other observer eavesdropping on channel $m'$ should know that the value of $m''$ is {\em not} equal to $10110$:
$$
r_3\Vdash\Box_m\Box_{m'}(m''\neq 10110).
$$

So far we have discussed epistemic properties of individual runs. A property which is true on one run does not have to be true on another. For example, the above formula $\Box_m\Box_{m'}(m''\neq 10110)$ is not true on any run in which the value of channel $m$ is $10110$. However, a similar property is true on all runs of protocol ${\cal P}_3$:
\begin{equation}\label{eq m' m''}
(m=01001)\to\Box_m\Box_{m'}(m''\neq 10110).
\end{equation}
Another property true for all runs of protocol $\mathcal{P}_3$ is
\begin{equation}\label{eq sym}
\Box_{m'}(m\neq 00000) \to \Box_{m'}(m''\neq 00000).
\end{equation}
Indeed, the assumption $\Box_{m'}(m\neq 00000)$ tells us that an observer of channel $m'$ on the current run can conclude that $m\neq 00000$. Since at most two mistakes can be introduced between channels $m$ and $m'$, we can conclude that the message that the observer sees on channel $m'$ contains at least three digits of 1. Therefore, for a similar reason, this observer will conclude that $m''\neq 00000$. 

A property true for all runs of one protocol does not have to be true for all runs of some other protocols. For example, property~(\ref{eq m' m''}) is false under protocol where up to two bits could be corrupted during each encryption and each decryption stage. Property~(\ref{eq sym}) is not true under a protocol where agents $q$ and $u$, unlike agents $s$ and $t$, are not allowed to make mistakes. 

In this article we study epistemic properties common to all protocols that have the same topological structure\footnote{As we formally define in the next section, the topological structure of a communication network is an undirected graph with multiple edges.} of communication networks. Consider, for example, property
\begin{equation}\label{eq gateway}
\Box_{m}(m''\neq 00000)\to \Box_{m'}(m''\neq 00000).
\end{equation}
We will see later in this article that this property is true for each protocol where, as in Figure~\ref{example3 figure}, communication between channels $m$ and $m''$ happens only through channel $m'$. 


The above formula~(\ref{eq gateway}) involves inequality. Neither inequality nor equality is a part of the language of our system. We only allow propositional symbols as atomic statements. An example of an epistemic property common to all protocols with the network topology depicted in Figure~\ref{example3 figure}  expressible in our language is:
\begin{equation}\label{box gateway}
\Box_{m}\Box_{m''}\phi\to\Box_{m'}\Box_{m''}\phi.
\end{equation}
Informally, this property states that if any observer eavesdropping on channel $m$ is able to deduce that any other observer eavesdropping on channel $m''$ can conclude that some property $\phi$ is true, then the same deduction can be made by any observer eavesdropping on channel $m'$ on the same run. This property, as shown in Example~\ref{example 3}, is a special case of our Gateway axiom. We prove the soundness of Gateway axiom with respect to a formally defined semantics in Section~\ref{soundness section}.

Another, perhaps surprising, example of a property common to all protocols with the network topology depicted in Figure~\ref{example3 figure} is:
\begin{equation}\label{box disjunction}
\Box_{m'}(\Box_{m}\phi\vee \Box_{m''}\psi)\to\Box_{m'}\Box_m\phi\vee\Box_{m'}\Box_{m''}\psi.
\end{equation}
Generally speaking, the knowledge of a disjunction of two formulas does not imply the knowledge of either of the two disjuncts. The above formula, however, states that this is true when the disjunct talks about the knowledge of observers located on different sides of channel $m'$. In Section~\ref{examples section}, we prove a more general form of property~(\ref{box disjunction}).

An epistemic logic for reasoning about communication graphs was proposed by Pacuit and Parikh~\cite{pp07ilgss}. Their language consists of two different modalities: an epistemic modality $K_a$ labeled by an agent $a$ and a modality $\Box$ interpreted as ``after any sequence of communications under the given protocol it is true that". They discussed logical principles specific to a given network topology and even gave, in the introduction, a principle similar to our Gateway axiom. However, they did not provide a complete axiomatization of their logical system for a specific topology, even though they proved its decidability.

Kane and Naumov~\cite{kn13tark} proposed a similar logical system whose language contains only epistemic modality. They eliminated modality ``after any sequence of communications" by assuming that all statements refer to the final knowledge after the communication. In this simplified setting they have been able to prove completeness theorem, but only for the case of linear communication networks.

This article extends Kane and Naumov's work from linear communication chains to arbitrary connected graphs. The logical system introduced in~\cite{kn13tark} contained two principles capturing topology of linear communication chains: Gateway axiom and Disjunction axiom, similar to properties~(\ref{box gateway}) and (\ref{box disjunction}) above. The more general version of Gateway axiom described in the current article no longer requires Disjunction property as a separate axiom. Instead, we prove this property from the more general version of Gateway axiom in Lemma~\ref{second vee lemma}. More importantly, the proof of the completeness theorem for non-linear graphs is completely different from the proof of completeness for linear communication chains. In the case of the proof of completeness for linear communication chains, if an observer of channel $m$ knows certain information about channel $m'$, then it is enough to simply pass this information along the interval between channels $m$ and $m'$. However, the same technique does not apply to non-linear graphs. As we have demonstrated with protocol $\mathcal{P}_1$ and properties~(\ref{c not knows}) and (\ref{k not knows}), in non-linear graphs an observer of channel $m$ might know certain information about channel $m'$ without anyone between them knowing this information. To be able to prove completeness for non-linear graphs we introduce a new {\em network flow} construction described in Section~\ref{completeness section}.

An applied value of the result in this article is in providing a uniform protocol design procedure for communication networks. Namely, suppose that one needs to design a protocol for a network that satisfies security conditions $\phi_1$,\dots, $\phi_n$ expressed in our modal language. Assume additionally that the physical layout (topological structure) of the network is given and can not be changed. In such a setting, the protocol designer should be able to either (i) derive formula $\bigwedge_{i\le n}\phi_i\to\bot$ in our logical system  and, thus, prove that the specification of the protocol can not be met, or (ii) use the construction from our proof of completeness to produce a protocol that satisfies each of the desired conditions $\phi_1$,\dots, $\phi_n$.

Tao, Slutzki, and Honavar~\cite{tsh14tocl} introduced a conceptual logical framework for answering queries without revealing secrecy to multiple querying agents where there is a set of secrets that need to be protected against each of these agents. The communication between agents is modeled using a graph. The focus of their work is on a privacy-preserving algorithm, not on an axiomatic system. 
 
This article is also related to the works on information flow on graphs~\cite{dmn11wollic,hn12jelia,mn11amai,mn11clima,mn11apal,kn14synthese}, that study properties of nondeducibility, functional dependency, common knowledge, and fault tolerance predicates. Unlike those works, this article is using a modal language.

The article is organized as follows. Section~\ref{graph section} introduces relevant terminology from graph theory. Section~\ref{semantics section} defines the formal syntax and the semantics for our logical system, which is introduced in Section~\ref{axioms section}. Section~\ref{examples section} illustrates our logical system by giving several examples of formal proofs in this system. Some of these examples are used later in the proof of completeness. The soundness of the system is established in Section~\ref{soundness section}.  The rest of the article is dedicated to the proof of completeness in Section~\ref{completeness section}. The proof starts with an informal discussion of a network flow protocol. It continues to formalize the network flow protocol as a canonical communication protocol over the graph. Finally, multiple instances of the canonical protocol are aggregated together to show the completeness of the logical system. Section~\ref{conclusion section} concludes the article.

\section{Graph Theory Preliminaries}\label{graph section}

We study epistemic properties common to all protocols with the same topology of a channel network. Under such a protocol, multiple messages can be sent over the same channel. A {\em value of a channel} is the set of all messages communicated through the channel, possibly in both directions. We specify the network topology as an undirected graph in which vertices represent agents and edges represent communication channels between agents. In this section we introduce graph terminology used throughout the rest of the article.

Graph $(V,E)$ contains a set of vertices $V$ and a set of edges $E$ with an incidence relation between them. We allow loops and  multiple edges between the same pair of vertices. We write $e\in Edge(v_0,v_1)$ to state that edge $e\in E$ is one of (possibly multiple) edges between vertices $v_0\in V$ and $v_1\in V$. By $Inc(v)$ we denote the set of all edges incident to vertex $v\in V$. By $Inc(e)$ we denote the set consisting of the two ends of edge $e\in E$. For example, $Inc(q)=\{m,k,c\}$ and $Inc(k)=\{q,u\}$ in Figure~\ref{example3 figure}.

Let $e\in E$ be an edge of a graph $(V,E)$ incident to a vertex $v\in V$. If edge $e$ is removed from the graph, remaining graph $(V,E\setminus\{e\})$ might have up to two connected components. By $C^v_{\miniminus e}$ we denote the {\em connected component} of the graph $(V,E\setminus\{e\})$ that contains vertex $v$. Note that in some cases component $C^v_{\miniminus e}$ might be equal to the entire graph $(V,E\setminus\{e\})$. For the graph in Figure~\ref{example3 figure}, component $C^u_{\miniminus m'}$ consists of vertices $p$, $q$, and $u$ as well as edges $m$, $k$, and $c$. For the same graph, component $C^u_{\miniminus k}$ contains all vertices of the original graph and all edges of that graph except for edge $k$.

A {\em path}\label{simple path} is a sequence $e_0,v_1,e_1,\dots,v_k,e_k$ such that $k\ge 0$, $e_0,\dots,e_k$ are distinct edges, and $v_1,\dots,v_k$ are distinct vertices of the graph such that $e_i,e_{i+1}\in Inc(v_{i+1})$ for each $0\le i< k$. In Figure~\ref{example3 figure}, sequence $k,u,m',v,c'$ and one-element sequence $c$ are both examples of paths. A {\em circular path} is defined similarly except for edges $e_0$ and $e_k$ being the same. 

\begin{definition}
Edge $g$ is a gateway between sets of edges $A$ and $B$ of a graph if each path that starts with an edge in set $A$ and ends with an edge in set $B$ contains the edge $g$.
\end{definition}
For example, edge $m'$ is a gateway between sets of edges $\{m,k\}$ and $\{k',c'\}$ in Figure~\ref{example3 figure}. Note that in the above definition edge $g$ can belong to either or both of the sets $A$ and $B$. In Figure~\ref{example3 figure}, edge $k$ is a gateway between singleton set $\{k\}$ and set $\{m,m''\}$. 

\section{Syntax and Semantics}\label{semantics section}

In this section we define the language and the formal semantics of our logical system. These definitions presuppose a fixed {\em signature} of the communication network.

\begin{definition}\label{signature}
A signature $\Sigma$ is an arbitrary triple $\Sigma=(V,E,\{P_e\}_{e\in E})$, such that $(V,E)$ is a connected graph and $\{P_e\}_{e\in E}$ is a family of disjoint sets of propositions.
\end{definition}
Informally, propositions in set $P_e$ are atomic statements about values of the communication channel $e$. 

Different connected components of a disconnected graph can not exchange any information between them, so, for the sake of simplicity, we have chosen to restrict our system to connected graphs. 

We next define the language of our logical system. 

\begin{definition}
For every signature $\Sigma$, let $\Phi(\Sigma)$ be the minimal set of formulas such that
\begin{enumerate}
\item $\bot\in \Phi(\Sigma)$,
\item $P_e\subseteq \Phi(\Sigma)$ for every $e\in E$,
\item if $\phi,\psi\in \Phi(\Sigma)$, then $\phi\to\psi\in\Phi(\Sigma)$,
\item if $e\in E$ and $\phi\in\Phi(\Sigma)$, then $\Box_e\phi\in\Phi(\Sigma)$.
\end{enumerate}
\end{definition}
We assume that connectives $\neg$, $\wedge$, and $\vee$ are defined through $\rightarrow$ and $\bot$ in the usual way.

 Informally, a protocol is specified by giving a range of values\footnote{Each value represents the collection of all messages sent through the channel on a given run.} for each edge (``communication channel") and establishing dependencies between the values of the edges. These dependencies are ``enforced" by vertices (``agents"), and, thus, each such condition only involves edges incident to a vertex. For this reason we refer to these conditions as ``local". For example, for protocol $\mathcal{P}_1$ in the introduction, the local condition enforced by vertex $q$ is $c=m\oplus k$, where $m\oplus k$ is a {\em bit-wise exclusive or} of binary strings transmitted over channels $m$ and $k$. For protocols $\mathcal{P}_2$ and $\mathcal{P}_3$, the local condition at vertex $q$ is $h(c,m\oplus k)\le 1$, where $h(\cdot,\cdot)$ denotes the Hamming distance between any two binary strings of the same length. The local condition for vertex $p$ under all three of the above protocols is the constant {\em true}. In the formal definition below, a local condition is treated not as a Boolean function but rather as a set of tuples on which this function is true.
 
Recall that each atomic proposition $p$ in set $P_e$ is viewed as proposition ``about" the value of channel $e$. In what follows, by $\pi(p)$ we informally mean the set of all values of channel $e$ for which proposition $p$ is true.
 

\begin{definition}\label{protocol}
A protocol over a signature $(V,E,\{P_e\}_{e\in E})$ is a tuple $(\{W_e\}_{e\in E}$, $\{L_v\}_{v\in V}$, $\pi)$ such that
\begin{enumerate}
\item for every edge $e\in E$, set $W_e$ is an arbitrary set of values,
\item for every $v\in V$, set $L_v\subseteq \prod_{e\in Inc(v)}W_e$ specifies local conditions at vertex $v$,
\item for every $p\in P_e$, function $\pi$ is such that $\pi(p) \subseteq W_e$. We denote $\pi(p)$ by $p^\pi$.
\end{enumerate}
\end{definition}

\begin{definition}
A run of a protocol $(\{W_e\}_{e\in E},\{L_v\}_{v\in V}, \pi)$ is an arbitrary tuple $\langle w_e\rangle_{e\in E} \in \prod_{e\in E}W_e$ such that $\langle w_e\rangle_{e\in Inc(v)}\in L_v$ for every $v\in V$.
\end{definition}

\begin{definition}\label{tworun}
For any two tuples $r=\langle w_e\rangle_{e\in E}$ and $r'=\langle w'_e\rangle_{e\in E}$ and any $f\in E$, we write $r=_f r'$ if $w_f=w'_f$. 
\end{definition}

\begin{corollary}
Relation $r=_e r'$ is an equivalence relation. \qed
\end{corollary}

The formal semantics of our logical system is defined in terms of runs of a protocol, rather than in more common terms of epistemic worlds of a Kripke model. Note, however, that any protocol can be viewed as a Kripke model in which runs of the protocol are epistemic worlds and equality of runs on a given channel $c$ is the indistinguishability relation $\sim_c$ on epistemic worlds.

\begin{definition}\label{sat}
For every signature $\Sigma=(V,E,\{P_e\}_{e\in E})$, every $\phi\in\Phi(\Sigma)$, every protocol $\mathcal{P}=(\{W_e\}_{e\in E},\{L_v\}_{v\in V},\pi)$ over graph $(V,E)$, and every run $r=\langle w_e\rangle_{e\in E}$ of $\mathcal{P}$, relation $r\Vdash \phi$ is defined recursively as:
\begin{enumerate}
\item $r\nVdash\bot$,
\item $r\Vdash p$ if $w_e\in p^\pi$, where $p\in P_e$,
\item $r\Vdash \psi\to\chi$ if $r\nVdash \psi$ or $r\Vdash \chi$,
\item $r\Vdash \Box_e\psi$ if $r'\Vdash\psi$ for every run $r'$ of $\mathcal{P}$ such that $r'=_e r$. 
\end{enumerate}
\end{definition}

For any signature $\Sigma$ and any set of edges $T$, by $\Phi(\Sigma,T)$ we mean the set of all formulas in $\Phi(\Sigma)$ in which all outermost modalities are labeled only by edges in $T$ and all atomic propositions outside of scopes of all modalities belong to $\bigcup_{t\in T}P_t$. For example, $\Box_a\Box_b\phi\to\Box_c\psi\in \Phi(\Sigma,\{a,c\})$. Also, if $p\in P_a$ and $q\in P_b$, then $\Box_b p \to q\in \Phi(\Sigma,\{b\})$. We use this notation to state our Gateway axiom in the next section. Below is the formal definition of this notation.

\begin{definition}\label{phi(S,A)}
For every signature $\Sigma=(V,E,\{P_e\}_{e\in E})$ and every $T\subseteq E$, let $\Phi(\Sigma,T)$ be the minimal set of formulas such that
\begin{enumerate}
\item $\bot\in\Phi(\Sigma,T)$,
\item $P_t\subseteq \Phi(\Sigma,T)$ for every $t\in T$,
\item if $\phi,\psi\in \Phi(\Sigma,T)$, then $\phi\to\psi\in\Phi(\Sigma,T)$,
\item if $t\in T$ and $\phi\in\Phi(\Sigma)$, then $\Box_t\phi\in\Phi(\Sigma,T)$.
\end{enumerate}
\end{definition}
Note that in item 4 above, formula $\phi$ is an element of set $\Phi(\Sigma)$ rather than set $\Phi(\Sigma,T)$.

\section{Logical System}\label{axioms section}

In this section we specify the axioms and the inference rules of our logical system for a given signature $\Sigma=(V,E,\{P_e\}_{e\in E})$. Our logical system, in addition to propositional tautologies in language $\Phi(\Sigma)$, contains the following axioms:
\begin{enumerate}
\item Truth: $\Box_e\phi\to\phi$, where $\phi\in\Phi(\Sigma)$,
\item Positive Introspection: $\Box_e\phi\to\Box_e\Box_e\phi$,  where $\phi\in\Phi(\Sigma)$,
\item Negative Introspection: $\neg\Box_e\phi\to\Box_e\neg\Box_e\phi$,  where $\phi\in\Phi(\Sigma)$,
\item Distributivity: $\Box_e(\phi\to\psi)\to(\Box_e\phi\to\Box_e\psi)$,  where $\phi, \psi\in\Phi(\Sigma)$,
\item Gateway: $\Box_e(\phi\to\psi)\to(\phi\to\Box_g\psi)$, where $e\in A$, $\phi\in \Phi(\Sigma,A)$, $\psi\in \Phi(\Sigma,B)$, and edge $g$ is a gateway between sets of edges $A\subseteq E$ and $B\subseteq E$. 
\end{enumerate}
 Note that axioms of Truth, Positive Introspection, Negative Introspection, and Distributivity are identical to the corresponding axioms of multi-agent epistemic logic S5. Thus, our logical system can be viewed as an extension of S5 by Gateway axiom. 
 
 \begin{figure}[ht]
\begin{center}
\vspace{3mm}
\scalebox{.6}{\includegraphics{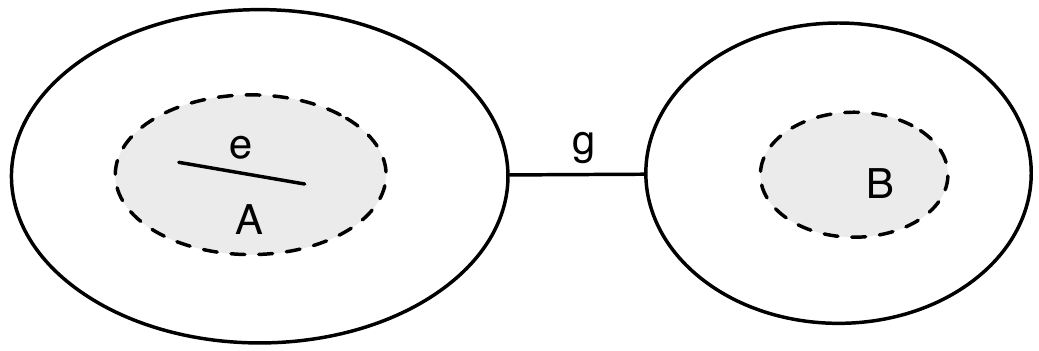}}
\vspace{0mm}
\caption{Edge $g$ is a gateway between sets of edges $A$ and $B$.}\label{gateway figure}
\vspace{0cm}
\end{center}
\vspace{0cm}
\end{figure}
 
Figure~\ref{gateway figure} illustrates the setting for Gateway axiom. To explain the intuition behind Gateway axiom, let us first consider the special case of this axiom when formula $\phi$ is a propositional tautology. In this case, Gateway axiom can be reduced to $\Box_e\psi\to\Box_g\psi$, which means that if an agent eavesdropping on channel $e$ knows something about the channels in set $B$, then an agent eavesdropping on gateway channel $g$ must also know this. Intuitively, this claim is true because the information about channels in set $B$ can only reach the observer of channel $e$ by flowing through the gateway channel $g$. However, to the best of our knowledge, Gateway axiom in this reduced form $\Box_e\psi\to\Box_g\psi$ does not yield a complete logical system. To achieve the completeness, we need a slightly more general principle that takes into account the ``local" information about channels on the same side of the gateway as channel $e$. In Gateway axiom $\Box_e(\phi\to\psi)\to(\phi\to\Box_g\psi)$ the local information is captured by formula $\phi$.

We write $\vdash_{Sig}\phi$ if formula $\phi$ is provable in our logical system for signature $Sig$ using Modus Ponens and Necessitation inference rules:
$$
\dfrac{\phi, \hspace{1cm} \phi\to\psi}{\psi}\hspace{2cm}\dfrac{\phi}{\Box_e\phi}
$$
where $\phi,\psi\in\Phi(\Sigma)$ and $e\in E$. We write $X\vdash_{Sig}\phi$ if formula $\phi$ is provable in our logical system from the set of assumptions $X$ using only Modus Ponens rule. We omit subscript $Sig$ when its value is clear from the context.

\section{Examples}\label{examples section}

The soundness and the completeness of our logical system will be established in the next two sections. In this section we give several examples of formal proofs in this system. Among these examples there are several lemmas that will be used later in the proof of completeness. 

\begin{figure}[ht]
\begin{center}
\vspace{3mm}
\scalebox{.6}{\includegraphics{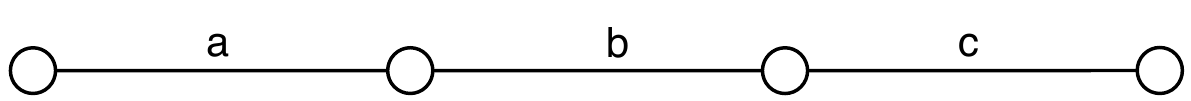}}
\vspace{0mm}
\caption{Three-Channel Linear Communication Network.}\label{example4 figure}
\vspace{0cm}
\end{center}
\vspace{0cm}
\end{figure}

\begin{example}
For any signature $\Sigma=(V,E,\{P_e\}_{e\in E})$ and any $\phi\in \Phi(\Sigma)$ where $(V,E)$ is the graph depicted in Figure~\ref{example4 figure}, 
$$\vdash_{Sig}\Box_a(\Box_b\phi\vee \Box_c\phi)\to\Box_b\phi.$$ 
In other words, if an observer eavesdropping on channel $a$ knows that an observer eavesdropping on channel $b$ knows $\phi$ or an observer eavesdropping on channel $c$ knows $\phi$, then the observer eavesdropping on channel $b$ must know $\phi$. 
\end{example}
\begin{proof}
Formula $\Box_c\phi\to\phi$ is an instance of Truth axiom. Thus, by Necessitation inference rule, $\vdash\Box_b(\Box_c\phi\to\phi)$. Hence, by Distributivity axiom and Modus Ponens inference rule, 
\begin{equation}\label{example2 eq}
\vdash\Box_b\Box_c\phi\to\Box_b\phi.
\end{equation}
At the same time note that edge $b$ is a gateway between sets $\{a,b\}$ and $\{c\}$. Additionally, $\neg\Box_b\phi\in\Phi(\Sigma,\{a,b\})$ and $\Box_c\phi\in\Phi(\Sigma,\{c\})$. Thus, by Gateway axiom,
$
\vdash \Box_a(\neg\Box_b\phi\to\Box_c\phi)\to(\neg\Box_b\phi\to\Box_b\Box_c\phi).
$
Hence, using statement~(\ref{example2 eq}) and the laws of propositional logic,
$
\vdash \Box_a(\neg\Box_b\phi\to\Box_c\phi)\to(\neg\Box_b\phi\to\Box_b\phi).
$
Note that formula $(\neg\Box_b\phi\to\Box_b\phi)\to\Box_b\phi$ is a propositional tautology. Thus,
$
\vdash \Box_a(\neg\Box_b\phi\to\Box_c\phi)\to\Box_b\phi
$. Finally, recall that disjunction $\Box_b\phi\vee\Box_b\phi$ is an abbreviation for $\neg\Box_b\phi\to\Box_b\phi$. Therefore, $\vdash \Box_a(\Box_b\phi\vee\Box_c\phi)\to\Box_b\phi$.
\end{proof}

\begin{figure}[ht]
\begin{center}
\vspace{3mm}
\scalebox{.6}{\includegraphics{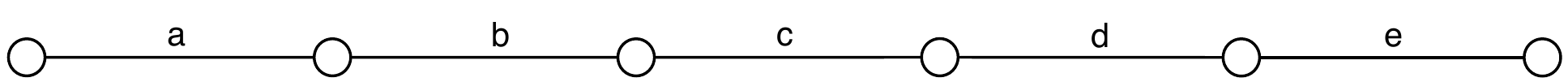}}
\vspace{0mm}
\caption{Five-Channel Linear Communication Network.}\label{example5 figure}
\vspace{0cm}
\end{center}
\vspace{0cm}
\end{figure}

In what follows, we denote by $\top$ the propositional tautology $\bot\to\bot$. 

\begin{example}
For any signature $\Sigma=(V,E,\{P_e\}_{e\in E})$ and any $\phi\in \Phi(\Sigma)$ where $(V,E)$ is the graph depicted in Figure~\ref{example5 figure}, 
$$\vdash_{\Sigma}\Box_a\Box_e\Box_c\phi\to\Box_b\Box_d\phi.$$ 
\end{example}
\begin{proof}
By Truth axiom, $\vdash\Box_c\phi\to\phi$. Thus, $\vdash\Box_d(\Box_c\phi\to\phi)$ by Necessitation inference rule. Hence, by Distributivity axiom and Modus Ponens rule,
\begin{equation}\label{eq dc}
\vdash\Box_d\Box_c\phi\to\Box_d\phi.
\end{equation}

At the same time, formula $\Box_c\phi\to(\top\to\Box_c\phi)$ is a propositional tautology. Thus, by Necessitation rule, $\vdash\Box_e(\Box_c\phi\to(\top\to\Box_c\phi))$. By Distributivity axiom and Modus Ponens inference rule,
\begin{equation}\label{eq ec}
\vdash\Box_e\Box_c\phi\to\Box_e(\top\to\Box_c\phi).
\end{equation}
Similarly, one can show that
\begin{equation}\label{eq ad}
\vdash \Box_a\Box_d\phi\to\Box_a(\top\to\Box_d\phi).
\end{equation}
Since edge $d$ is a gateway between the sets of edges $\{e\}$ and $\{c\}$, $\top\in\Phi(\Sigma,\{e\})$, and $\Box_c\phi\in\Phi(\Sigma,\{c\})$, by Gateway axiom,
$
\vdash\Box_e(\top\to\Box_c\phi)\to(\top\to\Box_d\Box_c\phi).
$
Hence, using statement~(\ref{eq dc}), statement~(\ref{eq ec}) and the propositional reasoning, 
$
\vdash\Box_e\Box_c\phi\to\Box_d\phi.
$
Thus, by Necessitation inference rule,
$
\vdash\Box_a(\Box_e\Box_c\phi\to\Box_d\phi).
$
Then, by Distributivity axiom and Modus Ponens inference rule,
\begin{equation}\label{eq ec2}
\vdash\Box_a\Box_e\Box_c\phi\to\Box_a\Box_d\phi.
\end{equation}
Since edge $b$ is a gateway between sets of edges $\{a\}$ and $\{d\}$, $\top\in\Phi(\Sigma,\{a\})$, and $\Box_d\phi\in\Phi(\Sigma,\{d\})$, by Gateway axiom,
$
\vdash\Box_a(\top\to\Box_d\phi)\to(\top\to\Box_b\Box_d\phi).
$
Therefore, using statement~(\ref{eq ad}), statement (\ref{eq ec2}), and the propositional reasoning,
$\vdash\Box_a\Box_e\Box_c\phi\to\Box_b\Box_d\phi$.
\end{proof}

We next prove formula~(\ref{box gateway}) stated in Section~\ref{introduction}. 

\begin{example}\label{example 3}
For any signature $\Sigma=(V,E,\{P_e\}_{e\in E})$ and any $\phi\in \Phi(\Sigma)$, where $G=(V,E)$ is the graph depicted in Figure~\ref{example3 figure}, 
$$\vdash_{Sig}\Box_m\Box_{m''}\phi\to\Box_{m'}\Box_{m''}\phi.$$ 
\end{example}
\begin{proof} 
Formula $\Box_{m''}\phi \to (\top \to \Box_{m''}\phi)$  
is a propositional tautology in language $\Phi(\Sigma)$. 
Thus, by Necessitation inference rule, we have 
$\vdash \Box_{m}(\Box_{m''}\phi \to (\top \to \Box_{m''}\phi))$. 
By Distributivity axiom and Modus Ponens inference rule,
\begin{equation}\label{example1 eq}
\vdash \Box_{m}(\Box_{m''}\phi) \to \Box_{m}(\top \to \Box_{m''}\phi)).
\end{equation}

Note now that edge $m'$ is a gateway between sets of edges $\{m\}$ and $\{m''\}$. Also, $\top\in \Phi(\Sigma,\{m\})$ and $\Box_{m''}\phi\in \Phi(\Sigma,\{m''\})$. Thus, by Gateway axiom,
$
\vdash_G \Box_{m}(\top \to \Box_{m''}\phi)\to(\top \to \Box_{m'}\Box_{m''}\phi).
$
Hence, using statement~(\ref{example1 eq}), by the laws of propositional logic,
$\vdash_G \Box_{m}\Box_{m''}\phi\to(\top \to \Box_{m'}\Box_{m''}\phi).$
Therefore, again using propositional logic,
$
\vdash_G \Box_{m}\Box_{m''}\phi\to \Box_{m'}\Box_{m''}\phi.
$
\end{proof}

Instead of proving property~(\ref{box disjunction}) from the introduction, in Lemma~\ref{second vee lemma} we prove a slightly more general statement that later will be used in the proof of completeness. The proof of Lemma~\ref{second vee lemma} relies on the  following auxiliary lemma. Figure~\ref{gateway figure} illustrates the settings of both of these lemmas. 

\begin{lemma}\label{vee lemma}
$\vdash \Box_e(\phi\vee\psi)\to(\phi\vee\Box_g\psi)$, where
 edge $g$ is a gateway between sets of edges $A$ and $B$, $e\in A$, $\phi\in \Phi(\Sigma,A)$, and $\psi\in \Phi(\Sigma,B)$.
\end{lemma}
\begin{proof}
Recall that $\phi\vee\psi$ is an abbreviation for $\neg\phi\to\psi$. Thus, we need to show that $\vdash \Box_e(\neg\phi\to\psi)\to(\neg\phi\to\Box_g\psi)$, which is an instance of Gateway axiom. 
\end{proof}

\begin{lemma}\label{second vee lemma}
$\vdash \Box_g(\phi\vee\psi\vee\chi)\to(\phi\vee\Box_g\psi\vee\Box_g\chi)$, where edge $g$ is a gateway between sets $A$ and $B$, $\phi\in\Phi(\Sigma,\{g\})$,  $\psi\in\Phi(\Sigma,A)$, and $\chi\in\Phi(\Sigma,B)$.
\end{lemma}
\begin{proof}
Note first that $g$ is a gateway between sets $A\cup\{g\}$ and $B$. Thus, by Lemma~\ref{vee lemma}, 
$$
\vdash \Box_g(\phi\vee\psi\vee\chi)\to\phi\vee\psi\vee\Box_g\chi.
$$
Hence, by the laws of propositional logic,
$$
\vdash \Box_g(\phi\vee\psi\vee\chi)\to\phi\vee\Box_g\chi\vee\psi.
$$
By Necessitation inference rule,
$$
\vdash \Box_g(\Box_g(\phi\vee\psi\vee\chi)\to\phi\vee\Box_g\chi\vee\psi).
$$
By Distributivity axiom and Modus Ponens rule,
$$
\vdash \Box_g\Box_g(\phi\vee\psi\vee\chi)\to\Box_g(\phi\vee\Box_g\chi\vee\psi).
$$
By Positive Introspection axiom,
\begin{equation}\label{tuesday2}
\vdash \Box_g(\phi\vee\psi\vee\chi)\to\Box_g(\phi\vee\Box_g\chi\vee\psi).
\end{equation}
Second, note that edge $g$ is also a gateway between sets $\{g\}$ and $A$. Thus, again by Lemma~\ref{vee lemma},
$$
\vdash\Box_g(\phi\vee\Box_g\chi\vee\psi)\to \phi\vee\Box_g\chi\vee\Box_g\psi.
$$
Hence, taking into account statement~(\ref{tuesday2}),
$$
\vdash \Box_g(\phi\vee\psi\vee\chi)\to\phi\vee\Box_g\chi\vee\Box_g\psi,
$$
which by the laws of propositional logic is equivalent to
$$
\vdash \Box_g(\phi\vee\psi\vee\chi)\to\phi\vee\Box_g\psi\vee\Box_g\chi.
$$
\end{proof}

Next, we continue with two more auxiliary lemmas.  Lemma~\ref{XYZ} is also  used in the proof of completeness. Lemma~\ref{pre-XYZ} is referred to in the proof of Lemma~\ref{XYZ}.

\begin{lemma}\label{pre-XYZ}
$\vdash\phi\to\Box_e\varphi$ for each $\phi\in \Phi(\Sigma,\{e\})$. 
\end{lemma}

\begin{proof}
Formula $\phi\to\phi$ is a tautology. Thus, by Necessitation inference rule, $\vdash\Box_e(\phi\to\phi)$. Note that $e$ is a gateway between sets $\{e\}$ and $\{e\}$. By Gateway axiom, $\vdash\Box_e(\phi\to\phi)\to(\phi\to\Box_e\phi)$. Therefore, $\vdash \phi\to\Box_e\phi$.
\end{proof}

\begin{lemma}\label{XYZ}
If $X\subseteq \Phi(\Sigma,\{e\})$ and $\phi\in\Phi(\Sigma)$, then $X\vdash\phi$ implies $X\vdash\Box_e\phi$.
\end{lemma}
\begin{proof}
Suppose that $X\subseteq \Phi(\Sigma,\{e\})$ and $X\vdash\phi$ where $\phi\in\Phi(\Sigma)$, then there is a finite subset $\{\psi_1, \psi_2, \dots, \psi_n\}$ of $X$ such that $\psi_1, \psi_2, \dots, \psi_n\vdash \phi$. Hence, by Deduction theorem for propositional logic, we have $\vdash \psi_1\to (\psi_2\to\dots(\psi_n\to\phi)\dots)$. By Necessitation rule, $\vdash \Box_e(\psi_1\to (\psi_2\to\dots(\psi_n\to\phi)\dots))$. Applying Distributivity axiom and Modus Ponens $n$ times, we have 
$\Box_e\psi_1,\Box_e\psi_2,\dots,\Box_e\psi_n\vdash\Box_e\phi$. Hence, by Lemma~\ref{pre-XYZ}, $\psi_1,\psi_2,\dots,\psi_n\vdash\Box_e\phi$. Therefore, $X\vdash\Box_e\phi$.
\end{proof}

\section{Soundness}\label{soundness section}

In this section we prove the soundness of our logical system with respect to runs of a protocol $\cal P$ over a signature $\Sigma=(V,E,\{P_e\}_{e\in E})$. The soundness of propositional tautologies and Modus Ponens inference rule is straightforward. Below we prove the soundness of Necessitation inference rule and of each axiom as a separate lemma. 

\begin{lemma}[Necessitation]
If $e\in E$ and $r\Vdash\phi$ for each run $r$ of protocol $\cal P$, then $r\Vdash\Box_e\phi$ for each run $r$ of protocol $\cal P$.
\end{lemma}
\begin{proof}
Let $r$ be a run of protocol $\cal P$. To show that $r\Vdash\Box_e\phi$, consider any run $r'$ of protocol $\cal P$ such that $r'=_e r$. It is sufficient to prove that $r'\Vdash\phi$, which is true due to the assumption of the lemma.  
\end{proof}

\begin{lemma}[Truth]\label{truth}
For every $e\in E$, every formula $\phi\in\Phi(\Sigma)$, and every run $r$  of protocol $\cal P$,
if $r\Vdash\Box_e\phi$, then $r\Vdash\phi$.
\end{lemma}

\begin{proof}
Assume that $r\Vdash\Box_e\phi$. Thus, by Definition~\ref{sat}, $r'\Vdash\phi$ for every run $r'$ of protocol $\cal P$ such that $r'=_e r$. In particular, $r\Vdash\phi$. 
\end{proof}

\begin{lemma}[Positive Introspection]
For every $e\in E$, every formula $\phi\in\Phi(\Sigma)$, and every run $r$ of protocol $\cal P$,
if $r\Vdash\Box_e\phi$, then $r\Vdash\Box_e\Box_e\phi$.
\end{lemma}

\begin{proof}
Assume that $r\Vdash\Box_e\phi$. Let $r'$ be any run of protocol $\mathcal{P}$ such that $r'=_e r$. We need to show that $r'\Vdash\Box_e\phi$. Consider any run $r''$ of protocol $\mathcal{P}$ such that $r''=_e r'$. We need to show that $r''\Vdash\phi$. Indeed, $r''=_e r'=_e r$ due to the choice of $r'$ and $r''$. Hence, $r''\Vdash\phi$ by the assumption $r\Vdash\Box_e\phi$.
\end{proof}

\begin{lemma}[Negative Introspection]
For every $e\in E$, every formula $\phi\in\Phi(\Sigma)$, and every run $r$ of protocol $\cal P$,
if $r\Vdash\neg\Box_e\phi$, then $r\Vdash\Box_e\neg\Box_e\phi$.
\end{lemma}

\begin{proof}
Assume that $r\Vdash\neg\Box_e\phi$. Then there is a run $r'$ of protocol $\cal P$ such that $r'=_e r$ and $r'\nVdash\phi$. Consider now any run $r''$ of protocol $\mathcal{P}$ such that $r''=_e r$. It is sufficient to show that $r''\Vdash \neg\Box_e\phi$, which is true because $r'=_e r =_e r''$ and $r'\nVdash\phi$.
\end{proof}

The proof of the soundness of Gateway axiom relies on the following technical lemma.

\begin{lemma}\label{two runs}
For every set $F\subseteq E$, 
every formula $\phi\in\Phi(\Sigma,F)$, and every two runs $r$ and 
$r'$ of protocol $\mathcal{P}$, if 
$r=_e r'$ for all $e\in F$, then
$r\Vdash\phi$ if and only if $r'\Vdash\phi$.
\end{lemma}
\begin{proof}
We prove this by induction on the structural complexity of formula $\phi$. 
The base case is when $\phi$ is a propositional variable $p\in P_e$ for some $e\in E$. 
By Definition~\ref{sat}, $r\Vdash p$ is equivalent to $w_e\in p^\pi$, 
which, due to $w_e=w'_e$, in turn is equivalent to $w'_e\in p^\pi$. 
The latter is equivalent to $r'\Vdash p$, again by Definition~\ref{sat}.

The induction step involves the following cases:
\begin{enumerate}
\item  Suppose that $\phi$ is of the form $\neg \psi$. By Definition~\ref{sat}, $r\Vdash \phi$ is equivalent to $r\nVdash \psi$. By the induction hypothesis, $r\nVdash \psi$ is equivalent to $r'\nVdash \psi$, which, by Definition~\ref{sat}, is equivalent to $r'\Vdash \neg \psi$. 

\item Suppose that $\phi$ is of the form $\psi\to\chi$. By Definition~\ref{sat}, 
$r\Vdash\psi\to\chi$ is equivalent to the disjunction of 
$r\nVdash\psi$ and $r\Vdash\chi$, which is equivalent to the disjunction of 
$r'\nVdash\psi$ and $r'\Vdash\chi$ by the induction hypothesis. 
The latter is equivalent to $r'\Vdash\psi\to\chi$ by Definition~\ref{sat}.

\item Suppose that $\phi$ is of the form $\Box_e\psi$. By Definition~\ref{sat}, 
$r\Vdash\Box_e\psi$ if and only if $r''\Vdash\psi$ for every $r''$ such that $r''=_e r'$. Since $r'=_e r$, the latter statement is equivalent to $r''\Vdash\psi$ for every $r''$ such that $r''=_e r'$. By Definition~\ref{sat}, the latter is equivalent to  $r'\Vdash\Box_e\psi$.
\end{enumerate}
\end{proof}

\begin{lemma}[Gateway]
For every run $r=\langle w_e \rangle_{e\in E}$ of protocol $\cal P$, every gateway $g$ between sets of edges $A$ and $B$, every $a\in A$, and every $\phi\in\Phi(\Sigma,A)$, $\psi\in\Phi(\Sigma,B)$, if $r\Vdash\Box_a(\phi\to\psi)$ and $r\Vdash\phi$, then $r\Vdash\Box_g\psi$.
\end{lemma}

\begin{proof}
Consider any run $r' = \langle  w'_e \rangle_{e\in E}$ of protocol ${\cal P}$ such that $r'=_g r$. It suffices to show that $r'\Vdash\psi$. Consider a graph $G'=(V,E\setminus\{g\})$. Due to the assumption that $g$ is a gateway $A$ and $B$, graph $G'$ consists of two connected components $C_A$ and $C_B$ such that all edges in set $A$ belong to the component $C_A$ and all edges in set $B$ belong to the component $C_B$. Let $r^+$ be a tuple $\langle w^+_e \rangle_{e\in E}$ such that
$$
w^+_e = 
\begin{cases}
w_e & \mbox{if $e\in C_A\cup \{g\}$,}\\
w'_e & \mbox{if $e\in C_B\cup \{g\}$}.
\end{cases}
$$
Note that tuple $r^+$ is well defined due to the assumption that $r'=_g r$.
\begin{claim}
Tuple $r^+$ is a run of protocol ${\cal P}$.
\end{claim}
\begin{proof}
We need to show that $r^+$ satisfies local conditions of protocol $\cal P$ at any vertex $v\in V$.
If $v\in C_A$, then $w^+_e=w_e$ for each $e\in Inc(v)$ by the choice of $\langle w^+_e \rangle_{e\in E}$. Hence, $\langle w^+_e\rangle_{e\in Inc(v)}=\langle w_e\rangle_{e\in Inc(v)}\in L_v$. The case $v\in C_B$ is similar.
\end{proof}

We are ready to finish the proof of the lemma. Note that $r^+=_a r$ by the choice of $\langle w^+_e \rangle_{e\in E}$ and the assumption $a\in A$. Thus, $r^+\Vdash \phi\to\psi$ by the assumption $r\Vdash\Box_a(\phi\to\psi)$. At the same time, $r^+\Vdash \phi$ by Lemma~\ref{two runs} and the assumption $r\Vdash\phi$. Hence, $r^+\Vdash\psi$ by Definition~\ref{sat}. Therefore, $r'\Vdash\psi$ by the same Lemma~\ref{two runs} and the assumption $\psi\in\Phi(\Sigma,B)$.
\end{proof}

\section{Completeness}\label{completeness section}

In this section we prove the completeness of our logical system with respect to the formal semantics defined in Section~\ref{semantics section}. 

In general, to prove a completeness theorem for a logical system, for any statement not provable in this system, one needs to describe how to construct a model in which this statement is false. 
In our case, for each formula $\phi$ not provable in our logical system, we construct a protocol (``Kripke model") and a run (``epistemic world") of this protocol on which formula $\phi$ is not satisfied. This protocol will be obtained by {\em aggregating} simpler {\em canonical} protocols. Each canonical protocol synchronizes information known to different observers. For example, if an observer $a$ knows that an observer $b$ knows $\psi$, then one of the canonical protocols guarantees that observer $b$ indeed knows $\psi$. 

The construction of such canonical protocols is based on the network flow protocol~\cite[p.708]{clrc09}. Information flow has many properties similar to that of network flow. In fact, network flow is sometimes used to communicate information. For example, the hydraulic brake system in modern cars uses the flow of the brake fluid to communicate a braking signal from the brake pedal to the wheels. In a more general setting, one can consider a closed system of water pipes with several faucets and several sinks. If one of the faucets is pumping water into the system (somebody knows formula $\delta$), then at least one of the sinks must be leaking the water (forcing formula $\delta$ to be true). We will use such pipe systems to communicate information between different edges of the graph.

In this section we first informally discuss network flow protocols in more details. Next, we define ``canonical" protocols that formalize network flow protocol in the form needed for our proof of completeness. Finally, to finish the proof of completeness, we aggregate multiple canonical protocols into a single one.



\subsection{Network Flow Protocol}

Consider an example of a network of six pipes depicted in Figure~\ref{flow1 figure}. Assume that this network has two sink faucets located at edges $d$ and $f$. Furthermore, let us assume that
\begin{enumerate}
\item water can leak from the network only through faucets on edges $d$ and $f$,
\item water does not have to leak even if the faucet is open, and
\item all pipes can (but do not have to) add water into the system by pumping it in the middle of the pipes. 
\end{enumerate}
Throughout this section, atomic propositions $p$ and $q$ denote the statements ``faucet on the edge $d$ is open" and ``faucet on the edge $f$ is open", respectively. 

\begin{figure}[ht]
\begin{center}
\vspace{3mm}
\scalebox{.6}{\includegraphics{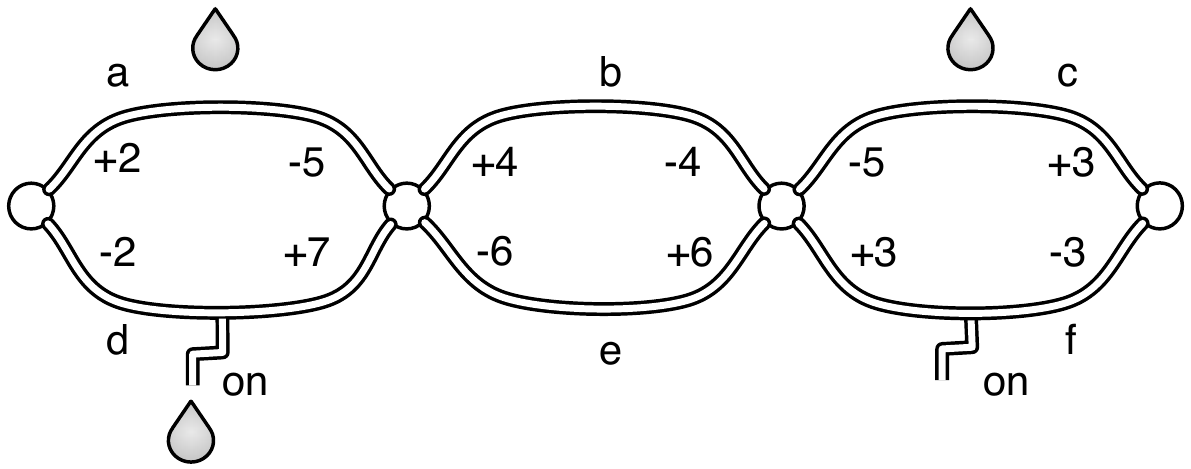}}
\vspace{0mm}
\caption{Run $r_1$ of a network flow protocol. }\label{flow1 figure}
\vspace{0cm}
\end{center}
\vspace{0cm}
\end{figure}

We show the flow in the network by assigning a real number $f^e_u$ to each end $u$ of each pipe $e$ in the network. The {\em positive} number denotes the speed (volume per time unit) with which water is {\em coming into} the pipe through this end and {\em negative} number shows the speed with which water is {\em leaving} the pipe through that end. 

So far, we assume that no water can be added at a vertex. Thus, the sum of all values at each vertex is zero. Any such valid assignment of the flow values to the ends of all pipes defines a run of the network flow protocol.

An example of a run $r_1$ is also shown on Figure~\ref{flow1 figure}. On this run pipes $a$ and $c$ add water into the system, both sink faucets are open, but only edge $d$ leaks water. Note that an external observer of pipe $a$ would see that the sum of flow values on edge $a$ is negative. This means that water is added into the system. Thus, the observer would be able to conclude that at least one of the sink faucets is open: $r_1\Vdash \Box_a(p\vee q)$. However, this observer will not be able to deduce exactly which faucet is open: $r_1\Vdash \neg\Box_a p\wedge \neg\Box_a q$. Also, an external observer of pipe $d$ will see that the sum of the two flow values at the ends of this pipe is positive and, thus, faucet on the pipe $d$ is leaking. Hence, $r_1\Vdash\Box_d p$ and so $r_1\Vdash\Box_d (p\vee q)$.

\begin{figure}[ht]
\begin{center}
\vspace{3mm}
\scalebox{.6}{\includegraphics{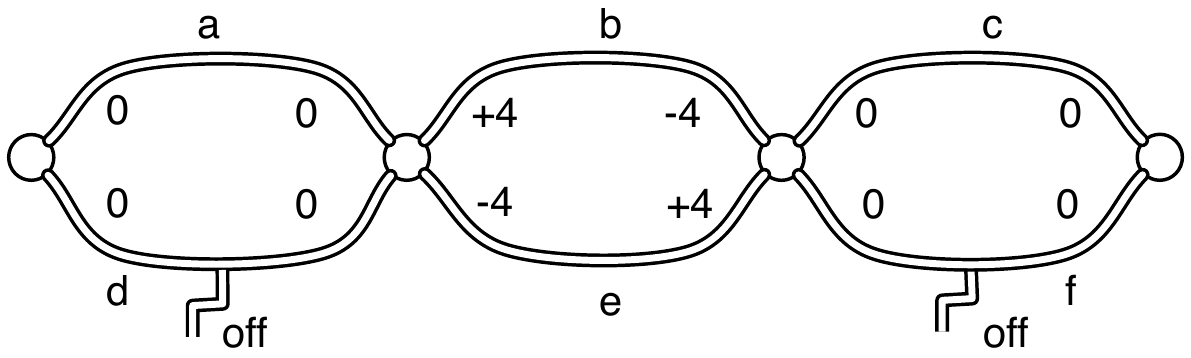}}
\vspace{0mm}
\caption{Run $r_2$ of a network flow protocol. }\label{flow2 figure}
\vspace{0cm}
\end{center}
\vspace{0cm}
\end{figure}

We now argue that $r_1\Vdash\neg\Box_b(p\vee q)$. Indeed, any external observer of pipe $b$ will not be able to distinguish run $r_1$ from run $r_2$ depicted in Figure~\ref{flow2 figure} because they have the same flow values at both ends of pipe $b$. Run $r_2$ has a circular flow through pipes $b$ and $e$, with both faucets being closed. Since $r_2\nVdash p\vee q$ and the observer of pipe $b$ can not distinguish between runs $r_1$ and $r_2$, it follows that $r_1\Vdash\neg\Box_b(p\vee q)$. 
Similarly, another run could be constructed to show that $r_1\Vdash\neg\Box_e (p\vee q)$.

Before continuing with the next example, let us introduce a notion of a {\em bridge} edge of a graph, which is related but not identical to the earlier introduced notion of a gateway edge between two sets of edges.

\begin{definition}\label{bridge}
An edge $b$ is a bridge in a connected graph $(V,E)$, if graph $(V,E\setminus\{b\})$ is not connected. 
\end{definition}
For any given graph, by $\cal B$ we mean the set of all bridges of this graph. For example, for the graph depicted in Figure~\ref{example3 figure}, set $\mathcal{B}$ is $\{m,m',m''\}$.

The main difference between a gateway and a bridge is that a gateway between sets is defined assuming two given sets. Bridge is a specific type of an edge. It's definition does not depend on the choice of any specific sets. Furthermore, a gateway does not have to be a bridge. For example, for any edges $e$ and $f$, of an arbitrary graph, edge $e$ is a gateway between set $\{e\}$ and set $\{f\}$ even if edge $e$ is not a bridge.

The graph in Figure~\ref{flow2 figure} has no bridges. As we show next, the epistemic properties of the network flow protocol are different for edges that are bridges and edges that are not bridges. Let $r_3$ be the run of the network flow protocol depicted in Figure~\ref{flow3 figure}, where pipe $b$ is a bridge. Note that although no additional water is pumped into pipe $b$, an external observer of pipe $b$ would be able to conclude that the faucet at edge $d$ is open because such an observer  would notice a right-to-left water flow on pipe $b$. In other words, $r_3\Vdash\Box_b p$.

\begin{figure}[ht]
\begin{center}
\vspace{3mm}
\scalebox{.6}{\includegraphics{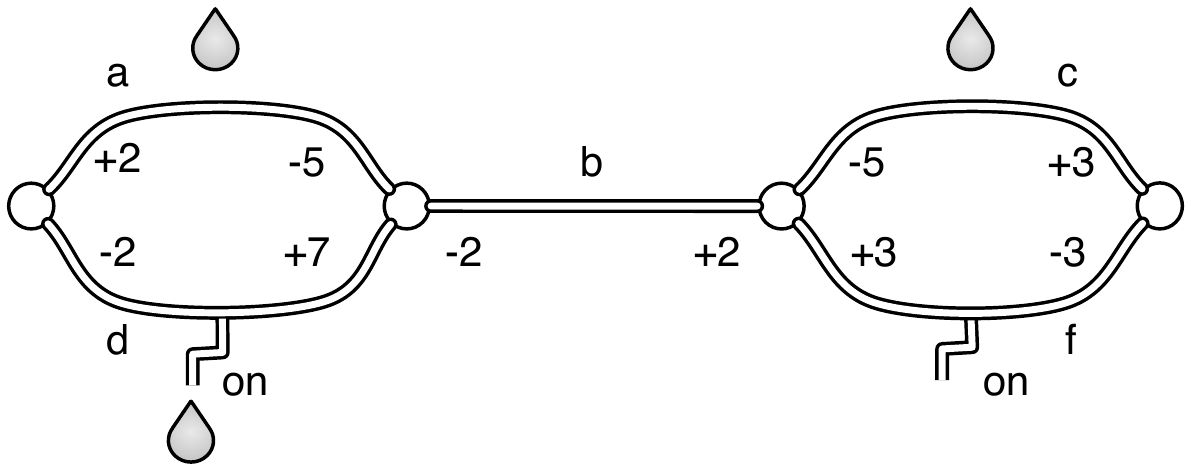}}
\vspace{0mm}
\caption{Run $r_3$ of a network flow protocol.}\label{flow3 figure}
\vspace{0cm}
\end{center}
\vspace{0cm}
\end{figure}

These examples show that in order for an observer of a non-bridge edge to be able to deduce disjunction $p\vee q$, this edge must be pumping water into the system. In the case over a bridge, however, it is sufficient to have a non-zero flow of the bridge in either of the two directions. This distinction between bridges and non-bridges under the network flow protocol will lead to two different corresponding cases in the definition of our canonical protocol (see Definition~\ref{values}).

The network flow protocol, as described above, has a peculiar property. Namely, since water could be pumped into the system only through edges, an external observer of bridge $b$ under run $r_3$ will not only be able to deduce that $p$ is true, but also to conclude that either an external observer of pipe $c$ or an external observer of pipe $f$ must know that $p\vee q$ is true: $r_3\Vdash \Box_b(\Box_c (p\vee q)\vee \Box_f (p\vee q))$. Indeed, an external observer of pipe $b$ would conclude that water is pumped into the system either at pipe $c$ or at pipe $f$ and, thus, either $\Box_c (p\vee q)$ or $\Box_f (p\vee q)$. To prove the completeness theorem for our logical system, we need a slightly more general class of flow protocols for which this property is not necessarily true. Namely, we allow additional water to be pumped into the system not only at pipes, but also at the vertices. The sink faucets, however, are still located only in the middle of the pipes. Under the modified network flow protocol, the statement $r_3\Vdash \Box_b(\Box_c (p\vee q)\vee \Box_f (p\vee q))$ is no longer true because an external observer of pipe $b$ can not distinguish run $r_3$ from run $r_4$ of the modified protocol depicted in Figure~\ref{flow4 figure} and because $r_4\Vdash\neg\Box_c(p\vee q)$ and $r_4\Vdash\neg\Box_f(p\vee q)$.

\begin{figure}[ht]
\begin{center}
\vspace{3mm}
\scalebox{.6}{\includegraphics{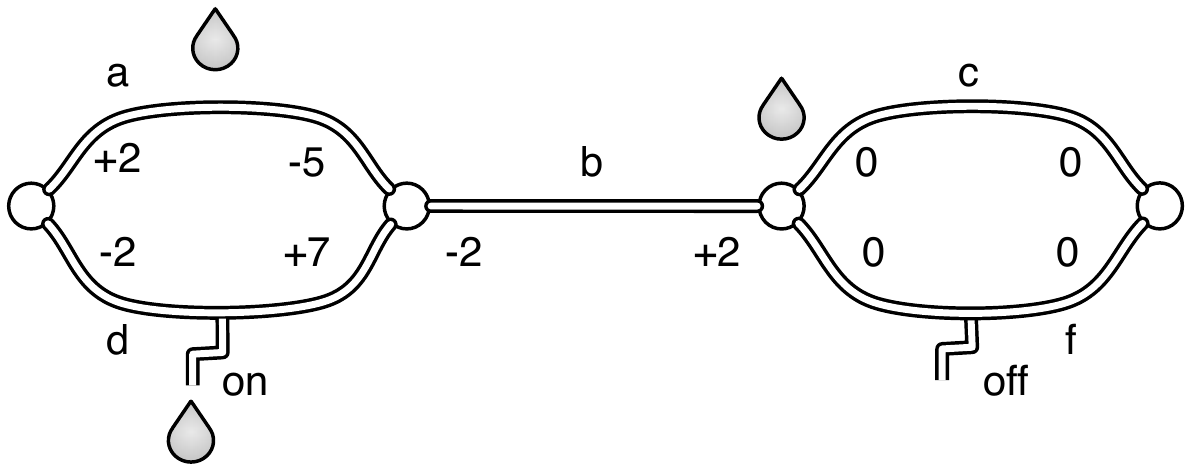}}
\vspace{0mm}
\caption{Run $r_4$ of a network flow protocol. }\label{flow4 figure}
\vspace{0cm}
\end{center}
\vspace{0cm}
\end{figure}

\subsection{Canonical Protocols}\label{canonical protocol section}

In this section we define canonical protocols based on the network flow construction informally discussed above. The canonical protocols are used later in the proof of completeness. Under a canonical protocol, the value of each edge $e$ contains a maximal consistent subset $X^e$ of $\Phi(\Sigma,\{e\})$. Informally, set $X^e$ consists of all epistemic facts about an external observer of edge $e$ that are true on a given run. Of course, on the same run, sets $X^e$ for different edges $e$ must be correlated. For example, if set $X^e$ contains formula $\Box_e\Box_h\psi$, then set $X^h$ must contain formula $\Box_h\psi$. In general, if $\Box_e\delta\in X^e$, then formula $\delta$ should be, in some sense, ``true" on this run. We use network flow to enforce such correlations between sets $X^e$ for different edges $e$ on the same run. 

A single canonical protocol is used to only enforce such a correlation for a single formula $\delta$. Thus, each formula $\delta$ produces a different canonical protocol. In Section~\ref{aggregation section}, we aggregate these canonical protocols into a single protocol. Note that in propositional logic any formula can be written in Disjunctive Normal Form. 
Any modal formula $\delta$ can be shown to be equivalent to $\bigwedge_{i\le n}\bigvee_{h\in E}\delta^i_h$, where $\delta^i_h\in \Phi(\Sigma,\{h\})$ for each $i\le n$ and each $h\in E$. Also note that in the presence of Distributivity axiom and Necessitation inference rule, formula $\Box_e\bigwedge_{i\le n}\bigvee_{h\in E}\delta^i_h$ is provably equivalent to $\bigwedge_{i\le n}\Box_e\bigvee_{h\in E}\delta^i_h$. Because of this, in what follows we enforce our correlation between different sets $X^e$ only for formulas $\delta$ of the form $\bigvee_{h\in E}\delta_h$, where $\delta_h\in \Phi(\Sigma,\{h\})$ for each $h\in E$. 

\begin{definition}\label{Delta}
For any signature $\Sigma=(V,E,\{P_e\}_{e\in E})$, let $\Delta(\Sigma)$ be the set of all formulas of the form $\bigvee_{e\in E}\delta_e$, where $\delta_e\in \Phi(\Sigma,\{e\})$ for each $e\in E$.
\end{definition}

The correlation that we intend to enforce is: for all $e\in E$, 
if $\Box_e\bigvee_{h\in E}\delta_h\in X^e$, then there exist $h\in E$ such that $\delta_h\in X^h$.
Instead of defining a single protocol ${\cal P}^\delta$ under which this correlation is enforced for each $e\in E$,
we define a family of protocols $\{{\cal P}_F^\delta\}_{F\subseteq E}$. For each subset $F\subseteq E$, under protocol ${\cal P}_F^\delta$ the correlation is enforced only for edges in $F$.  

The enforcement of the desired correlation under protocol ${\cal P}^\delta_F$ is achieved by using network the flow construction described in the previous section. Informally, each edge of the graph is viewed as a pipe. In addition to set $X^e$, the value of each edge $e$ also includes flow values over this edge. As before, sink faucets are placed in the middle of each edge. However, the sink faucet at edge $h$ is open only if $\delta_h\in X^h$. If $\Box_e\delta\in X^e$ and edge $e$ is {\em not} a bridge, then $e$ is required to ``pump" water into the system. The network flow protocol guarantees that if water is pumped into the system, then it must leak through at least one of the sinks. This implies that if $\Box_e\delta\in X^e$ (``water is pumped in"), then $\delta_h\in X^h$ (``sink is leaking") for at least one disjunct $\delta_h$  in formula $\delta$. For the same reason, if $\Box_e\delta\in X^e$ and $e$ is a bridge, then $e$ is required to have a non-zero flow (in either direction). 

We now define a canonical protocol ${\cal P}_F^\delta$ over a signature $\Sigma=(V,E,\{P_e\}_{e\in E})$ for each subset $F\subseteq E$ and each $\delta\in\Delta(\Sigma)$, where $\delta$ is of the form $\bigvee_{e\in E}\delta_e$ and $\delta_e\in \Phi(\Sigma,\{e\})$ for each $e\in E$.

\begin{definition}\label{values}
A value $w_e$ of an edge $e\in Edge(u,u')$ under protocol ${\cal P}_F^\delta$ is a tuple $\langle X, \{f_{v}\}_{v\in Inc(e)}\rangle$ that has the following properties:  

\begin{enumerate}
\item Properties common to all edges.
    \begin{enumerate}
    \item $X$ is a maximal consistent subset of $\Phi(\Sigma,\{e\})$,
    \item $f_u$ and $f_{u'}$ are real numbers, 
    \item $f_{u}+f_{u'}>0$ if and only if $\delta_e\in X$.
    \end{enumerate}
\item Properties of bridge edges. For each $e\in \cal B$,
    \begin{enumerate}
    \item if $\delta_e\notin X$, then $f_{u}+f_{u'}=0$,
    \item if $f_{u}<0$, then $\Box_e\bigvee_{h\in C^u_{\miniminus e}}\delta_h\in X$,
    \item if $e\in F$,  $\Box_e\delta\in X$, and $\delta_e\notin X$, then $f_{u}<0$ or $f_{u'}<0$.
    \end{enumerate}
\item Properties of non-bridge edges. For each $e\in E\setminus \cal B$,
    \begin{enumerate}
    \item if  $f_{u}+f_{u'}<0$, then $\Box_e\delta\in X$,
    \item if $e\in F$,  $\Box_e\delta\in X$, and $\delta_e\notin X$, then $f_{u}+f_{u'}<0$.\label{values line 5}
    \end{enumerate}
\end{enumerate}
\end{definition}

\paragraph{Valuation.} Let $\pi$ be a function such that, for each $e\in E$ and $p\in P_e$, set $p^\pi$ contains all values $\langle X, \{f_{v} \}_{v\in Inc(e)}\rangle$ under protocol $\mathcal{P}^\delta_F$,  where $p\in X$.   

We now specify local a condition $L_u$ at a vertex $u$ under  protocol ${\cal P}^\delta_F$. 
Under the network flow protocol, we allow any vertex $u$ to pump additional water into the system and disallow it to leak water out of the system. This is formally captured by the local condition $\sum_{e\in Inc(u)} f^e_{u} \ge  0$. At the same time, recall that we use the network flow to enforce property: if $\Box_e\delta\in X^e$, where $\delta=\bigvee_{h\in E}\delta_h$, then $\delta_h\in X^h$ for at least one $h\in E$. Note that if $\delta_h\in X^h$ for at least one $h\in E$, then the property is already true and no additional enforcement is necessary. Because of this, if $\delta_h\in X^h$ for at least one edge $h$ adjacent to vertex $u$, then we allow the sum $\sum_{e\in Inc(u)} f^e_{u}$ to be negative. This relaxation of the local condition will be useful later.

\paragraph{Local Conditions.}\label{local conditions} Consider any tuple of values 
$\langle X^{e }, \{f^{e}_{v} \}_{v\in Inc(e)} \rangle_{e\in Inc(u)}$
under protocol ${\cal P}^\delta_F$.
This tuple belongs to $L_u$ when the following condition is satisfied: if $\delta_e\notin X^e$ for each $e\in Inc(u)$, then $\sum_{e\in Inc(u)} f^e_{u} \ge  0$. 

This concludes the specification of the family of protocols ${\cal P}^\delta_F$. The following corollary directly follows from the above definitions. 

\begin{corollary}\label{scale}
For any run $\langle X^e, \{f^e_{u} \}_{u\in Inc(e)}\rangle_{e\in E}$ of a protocol ${\cal P}^\delta_F$ and any real number $\lambda>0$, tuple $\langle X^e, \{\lambda f^e_{u} \}_{u\in Inc(e)}\rangle_{e\in E}$ is a run of protocol ${\cal P}^\delta_F$.\qed
\end{corollary}

\begin{lemma}\label{gateway zero}
Let $\langle X^e, \{f^e_{u} \}_{u\in Inc(e)}\rangle_{e\in E}$ be any run of a protocol ${\cal P}^\delta_F$.  If $h = (v, v')\in F\cap \cal B$ and $\delta_h\notin X^h$, then $f^h_v=0$ if and only if $\Box_h\delta\notin X^h$.
\end{lemma}
\begin{proof}
$(\Rightarrow):$ We prove by contrapositive. Suppose that $\Box_h\delta\in X^h$. Then by Definition~\ref{values} part 2(c), $f^h_v<0$ or $f^h_{v'}<0$. Hence, $f^h_v\neq 0$ or $f^h_{v'}\neq 0$. Note that $f^h_{v'}\neq 0$, by Definition~\ref{values} part 2(a), implies that $f^h_v\neq 0$. Therefore, in both cases, $f^h_v\neq 0$.

\noindent$(\Leftarrow):$ 
Assume that $f^h_v\ne 0$. By Definition~\ref{values} part 2(a), either $f^h_v<0$ or $f^h_{v'}<0$. Suppose, without loss of generality, that $f^h_v<0$. Then, by Definition~\ref{values} part 2(b), 
\begin{equation}\label{anyname}
\Box_h\bigvee_{e\in C^v_{\miniminus h}}\delta_e\in X^h.
\end{equation}
Note that $\bigvee_{e\in C^v_{\miniminus h}}\delta_e\to\delta$ is a propositional tautology. Thus, by Necessitation rule, 
$$\vdash\Box_h\left(\bigvee_{e\in C^v_{\miniminus h}}\delta_e\to\delta\right).$$
Hence, by Distributivity axiom and Modus Ponens rule,
$$\vdash\Box_h\left(\bigvee_{e\in C^v_{\miniminus h}}\delta_e\right)\to\Box_h\delta.$$
Thus, $X^h\vdash\Box_h\delta$ from statement~(\ref{anyname}) and Modus Ponens inference rule. Therefore, $\Box_h\delta\in X^h$ due to the maximality of set $X^h$.
\end{proof}

\subsection{Properties of Canonical Protocols}

In this section we prove several technical properties of the canonical protocols that are used  in the proof of completeness. To build the intuition, as we proceed, we compare these properties with those of our informal network flow model.

\begin{lemma}\label{sub protocol}
For any $\delta\in \Delta(\Sigma)$, if $F'\subseteq F$, then each run of protocol ${\cal P}^\delta_F$ is also a run of protocol ${\cal P}^\delta_{F'}$.
\end{lemma}
\begin{proof}
The statement of the lemma immediately follows from the definition of the canonical protocols ${\cal P}^\delta_F$. Indeed, the difference between protocol ${\cal P}^\delta_F$ and ${\cal P}^\delta_{F'}$ is only in parts 2(c) and 3(b) of Definition~\ref{values}.
\end{proof}

The following theorem formalizes our intuition described earlier that if there is an inflow of water into the system, then there must be at least one open sink for the water to leak.

\begin{theorem}\label{XYZ2}
For any $h\in E$ and any run $\langle X^e, \{f^e_{u} \}_{u\in Inc(e)} \rangle_{e\in E}$ of protocol ${\cal P}^\delta_{\{h\}}$, if $\Box_h\delta\in X^h$, then there is an edge $h'\in E$ such that $\delta_{h'}\in X^{h'}$.
\end{theorem}
\begin{proof}
Suppose that there is no $h'\in E$ such that $\delta_{h'}\in X^{h'}$. 
Due to the local conditions of protocol ${\cal P}^\delta_{\{h\}}$,
\begin{equation}\label{saturday}
\sum_{e\in Inc(v)}f^e_v\ge 0, \hspace{5mm} \mbox{for each $v\in V$.}
\end{equation}
We consider the following two cases separately:

\noindent{\em Case I}:  $h\notin \cal B$. 
The sum of flow values over edges can be rearranged to the sum of flow values over vertices. 
Thus, due to inequality ~(\ref{saturday}),
\begin{equation}\label{ge zero}
\sum_{e\in Edge(u,u')}(f^e_u+f^e_{u'})=\sum_{v\in V}\sum_{e\in Inc(v)} f^e_v\ge 0.
\end{equation}
The assumption that there is no $h'\in E$ such that $\delta_{h'}\in X^{h'}$, together with the assumptions $h\notin \cal B$ and $\Box_h\delta\in X^h$ by part 3(b) of Definition~\ref{values}, implies that $f^h_v+f^h_{v'}<0$, where $v$ and $v'$ are the two ends of the edge $h$. Then, by inequality~(\ref{ge zero}), there must exist $h'\in Edge(u_1,u_2)$ such that $f^{h'}_{u_1}+f^{h'}_{u_2}>0$. Therefore, $\delta_{h'}\in X^{h'}$ by part 1(c) of Definition~\ref{values}, which is a contradiction.

\noindent{\em Case II}: $h\in \cal B$. By part 2(c) of Definition~\ref{values}, there is an end $u_0$ of edge $h$ such that $f^h_{u_0}<0$, see Figure~\ref{sink exists figure}. The sum of the flow values over edges in component $C^{u_0}_{\miniminus h}$ can be rearranged to the sum of the flow values over vertices. 
\begin{figure}[ht]
\begin{center}
\vspace{3mm}
\scalebox{.5}{\includegraphics{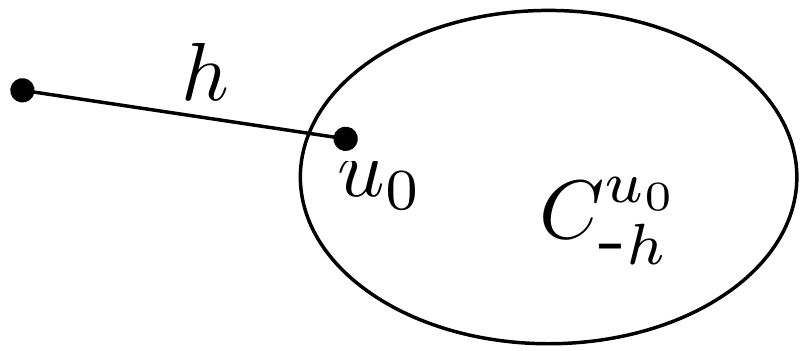}}
\vspace{0mm}
\caption{Towards proof of Theorem~\ref{XYZ2}, Case II.}\label{sink exists figure}
\vspace{0cm}
\end{center}
\vspace{0cm}
\end{figure}
Hence, by inequality~(\ref{saturday}),
\begin{equation*}
\sum_{e\in Edge(u,u')\cap C^{u_0}_{\miniminus h}}(f^e_u+f^e_{u'})=\left(\sum_{v\in C^{u_0}_{\miniminus h}}\sum_{e\in Inc(v)} f^e_v\right) - f^h_{u_0} \ge 0 - f^h_{u_0} > 0.
\end{equation*}
Thus, there must exist $h'\in Edge(u_1,u_2)\in C^{u_0}_{\miniminus h}$ such that $f^{h'}_{u_1}+f^{h'}_{u_2}>0$. Therefore, $\delta_{h'}\in X^{h'}$ by part 1(c) of Definition~\ref{values}, which is a contradiction.
\end{proof}

Note that in the network flow model the following property holds: if $v_1$ is one of the vertices of an edge $e_0$ and the water flows through edge $e_0$ towards vertex $v_1$, then there must exist a sink edge $e_k$ and a path $e_0,v_1,e_1,v_2,e_2,\dots,v_k,e_k$ such that there is a water flow along this path in the direction from edge $e_0$ to edge $e_k$. In our formal setting this property is captured by the following definition and lemma. 

\begin{definition}\label{gamma}
For any maximal consistent set of formulas $M$, let $\Gamma_M$ be the set of all paths $e_0,v_1,e_1,v_2,e_2,\dots,v_k,e_k$, where $k>0$, such that
\begin{enumerate}
\item $\Box_{e_0}\bigvee_{h\in C^{v_1}_{\miniminus e_0}}\delta_h\in M$,
\item $\delta_{e_i}\notin M$, for each $0\le i<k$, 
\item if $e_i\in \cal B$, then $\Box_{e_i}\left(\bigvee_{h\in C^{v_{i+1}}_{\miniminus e_i}}\delta_h\right)\in M$, for each $0\le i<k$,

\item $\delta_{e_k}\in M$.
\end{enumerate}
\end{definition}



\begin{lemma}\label{path exists}
For any edge $e\in Edge(u,u')$, if $\Box_{e}\bigvee_{h\in C^{u}_{\miniminus e}}\delta_h\in M$ and $\delta_e\notin M$, then there is a path in set $\Gamma_M$ that starts with edge $e$ and continues through vertex $u$. 
\end{lemma}

\begin{proof}
Let $\Omega$ be the set of all such paths $e_0,v_1,e_1,v_2,e_2,\dots,v_k,e_k$ that $e_0=e$, $v_1=u$, $\Box_{e_0}\delta\in M$, and for each $0\le i<k$, if $e_i\in \cal B$, then $\Box_{e_i}\left(\bigvee_{h\in C^{v_{i+1}}_{\miniminus e_i}}\delta_h\right)\in M$. 

Let $C_0$ be the set of all edges that belong to at least one path in $\Omega$. Let $C_1,\dots,C_n$ be the connected components of the graph obtained from component $C_{\miniminus e}^u$ by removing all edges in $C_0$. By the definition of set $\Omega$, for each $0<i\le n$ there is an edge $g_i$ in $C_0\cap \cal B$, such that 
\begin{equation}\label{bridge eq}
\Box_{g_i}\left(\bigvee_{h\in C_i}\delta_h\right)\notin M.
\end{equation}
Note that edge $g_i$ is the gateway between edges in $C_0\cup C_1\cup \dots \cup C_{i-1}\cup C_{i+1}\cup\dots \cup C_n$ and $C_{i}$. See Figure~\ref{path exists figure}.

\begin{figure}[ht]
\begin{center}
\vspace{3mm}
\scalebox{.5}{\includegraphics{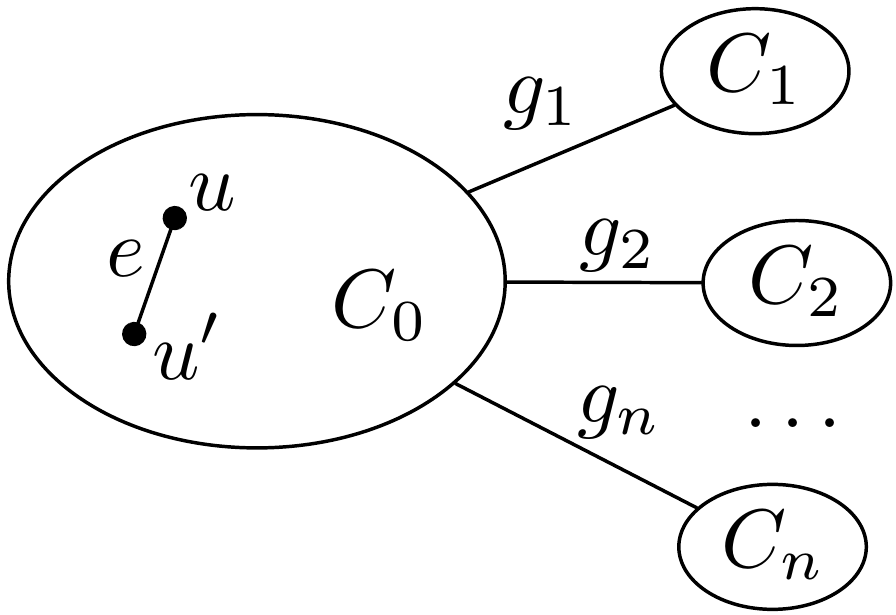}}
\vspace{0mm}
\caption{Components and corresponding bridges.}\label{path exists figure}
\vspace{0cm}
\end{center}
\vspace{0cm}
\end{figure}

The following formula is a propositional tautology:
$$
\left(\bigvee_{h\in C_{\miniminus e}^u}\delta_h\right) \to \left(\bigvee_{h\in C_0}\delta_h\right) \vee \left(\bigvee_{i=1}^n\bigvee_{h\in C_i}\delta_h\right).
$$
Thus, by Necessitation inference rule,
$$
\vdash \Box_e\left (\left(\bigvee_{h\in C_{\miniminus e}^u}\delta_h\right) \to \left(\bigvee_{h\in C_0}\delta_h\right) \vee \left(\bigvee_{i=1}^n\bigvee_{h\in C_i}\delta_h\right)\right).
$$
By Distributivity axiom,
$$
\vdash \Box_e\left(\bigvee_{h\in C_{\miniminus e}^u}\delta_h\right) \to \Box_e\left(\left(\bigvee_{h\in C_0}\delta_h\right) \vee \left(\bigvee_{i=1}^n\bigvee_{h\in C_i}\delta_h\right)\right).
$$
By Lemma~\ref{vee lemma} and laws of propositional logic,
$$
\vdash \Box_e\left(\bigvee_{h\in C_{\miniminus e}^u}\delta_h\right) \to  \left(\left(\bigvee_{h\in C_0}\delta_h\right) \vee \left(\bigvee_{i=2}^n\bigvee_{h\in C_i}\delta_h\right)\right)\vee \Box_{g_1}\left(\bigvee_{h\in C_1}\delta_h\right).
$$
By Necessitation rule,
$$
\vdash \Box_e\left(\Box_e\left(\bigvee_{h\in C_{\miniminus e}^u}\delta_h\right) \to  \left(\left(\bigvee_{h\in C_0}\delta_h\right) \vee \left(\bigvee_{i=2}^n\bigvee_{h\in C_i}\delta_h\right)\right)\vee \Box_{g_1}\left(\bigvee_{h\in C_1}\delta_h\right)\right).
$$
By Distributivity axiom,
$$
\vdash \Box_e\Box_e\left(\bigvee_{h\in C_{\miniminus e}^u}\delta_h\right) \to  \Box_e\left(\left(\left(\bigvee_{h\in C_0}\delta_h\right) \vee \left(\bigvee_{i=2}^n\bigvee_{h\in C_i}\delta_h\right)\right)\vee \Box_{g_1}\left(\bigvee_{h\in C_1}\delta_h\right)\right).
$$
By Positive Introspection axiom,
$$
\vdash \Box_e\left(\bigvee_{h\in C_{\miniminus e}^u}\delta_h\right) \to  \Box_e\left(\left(\left(\bigvee_{h\in C_0}\delta_h\right) \vee \left(\bigvee_{i=2}^n\bigvee_{h\in C_i}\delta_h\right)\right)\vee \Box_{g_1}\left(\bigvee_{h\in C_1}\delta_h\right)\right).
$$
By Lemma~\ref{vee lemma} and the laws of propositional logic,
$$
\vdash \Box_e\left(\bigvee_{h\in C_{\miniminus e}^u}\delta_h\right) \to  \left(\left(\bigvee_{h\in C_0}\delta_h\right) \vee \left(\bigvee_{i=3}^n\bigvee_{h\in C_i}\delta_h\right)\vee \left(\bigvee_{i=1}^2\Box_{g_i}\left(\bigvee_{h\in C_i}\delta_h\right)\right)\right).
$$
By repeating the previous steps $n-2$ more times,
$$
\vdash \Box_e\left(\bigvee_{h\in C_{\miniminus e}^u}\delta_h\right) \to  \left(\left(\bigvee_{h\in C_0}\delta_h\right) \vee \left(\bigvee_{i=1}^n\Box_{g_i}\left(\bigvee_{h\in C_i}\delta_h\right)\right)\right).
$$
Since, $\Box_e\left(\bigvee_{h\in C_{\miniminus e}^u}\delta_h\right)\in M$ and set $M$ is a maximal consistent set of formulas,
$$
\left(\bigvee_{h\in C_0}\delta_h\right) \vee \left(\bigvee_{i=1}^n\Box_{g_i}\left(\bigvee_{h\in C_i}\delta_h\right)\right)\in M.
$$
Due to (\ref{bridge eq}) and the maximality of set $M$, there must exist an edge $h\in C_0$ such that $\delta_h\in M$. By the definition of $C_0$, there is a path $e,v_1,e_1,v_2,e_2,\dots,v_k,e_k$ in $\Omega$ containing $h$. Let $e_m$ be the first edge along this path such that $\delta_{e_m}\in M$. Note that $e_m\neq e$ because $\delta_e\notin M$ by the assumption of the claim. Then, $e,v_1,e_1,v_2,e_2,\dots,v_m,e_m$ is the required path in $\Gamma$. 
\end{proof}

Another property that holds for the network flow is: if water is pumped into an edge $e_0$, then there must exist a sink edge $e_k$ and a path $e_0,v_1,e_1,\dots,v_k,e_k$ such that there is a water flow along this path in the direction from edge $e_0$ to edge $e_k$. We capture this property in the canonical protocol case by the following lemma. 

\begin{lemma}\label{path exists 2}
For any edge $e\in E$, and any $\delta\in\Delta(\Sigma)$, if $\Box_e\delta\in M$ and $\delta_e\notin M$, then there is a path in set $\Gamma_M$ that starts with edge $e$. 
\end{lemma}
\begin{proof}
Let $e\in Edge(u,u')$. There are two cases:

\noindent{\em Case I:} $e\in E\setminus \cal B$.  Note that $e$ is a gateway between sets $\{e\}$ and $E\setminus\{e\}$. Then, by Lemma~\ref{vee lemma}, 
\begin{equation}\label{tuesday}
\vdash \Box_e\left(\delta_e\vee \bigvee_{h\in C^u_{\miniminus e}}\delta_h\right)\to\left(\delta_e\vee\Box_e\bigvee_{h\in C^u_{\miniminus e}}\delta_h\right)
\end{equation}
At the same time, component $C^u_{\miniminus e}$ contains all edges of the graph except for edge $e$ due to the assumption $e\in E\setminus \cal B$. Thus,
$$
\delta\to\delta_e\vee \bigvee_{h\in C^u_{\miniminus e}}\delta_h
$$
is a propositional tautology. Hence, by Necessitation inference rule,
$$
\vdash \Box_e\left(\delta\to\delta_e\vee \bigvee_{h\in C^u_{\miniminus e}}\delta_h\right).
$$
By Distributivity axiom and Modus Ponens inference rule,
$$
\vdash \Box_e\delta\to\Box_e\left(\delta_e\vee \bigvee_{h\in C^u_{\miniminus e}}\delta_h\right).
$$
Using statement~(\ref{tuesday}) and the laws of propositional logic,
$$
\vdash \Box_e\delta\to \delta_e\vee\Box_e\bigvee_{h\in C^u_{\miniminus e}}\delta_h.
$$
Recall that $\Box_e\delta\in M$ and $\delta_e\notin M$. Thus, $\Box_e\bigvee_{h\in C^u_{\miniminus e}}\delta_h\in M$, due to the maximality and the consistency of set $M$. Then, the required follows from Lemma~\ref{path exists}.

\noindent{\em Case II:} $e\in \cal B$. Thus, edge $e$ is a gateway between edges of the component $C^u_{\miniminus e}$ and edges of the component $C^{u'}_{\miniminus e}$. Thus, by Lemma~\ref{second vee lemma}, 
\begin{equation}\label{tuesday3}
\vdash \Box_e\left(\delta_e\vee \bigvee_{h\in C^u_{\miniminus e}}\delta_h\vee \bigvee_{h\in C^{u'}_{\miniminus e}}\delta_h\right) 
\to\left(\delta_e\vee\Box_e\bigvee_{h\in C^u_{\miniminus e}}\delta_h \vee\Box_e\bigvee_{h\in C^{u'}_{\miniminus e}}\delta_h\right).
\end{equation}
At the same time, notice that the formula
$$
\delta\to\delta_e\vee \bigvee_{h\in C^u_{\miniminus e}}\delta_h
\vee \bigvee_{h\in C^{u'}_{\miniminus e}}\delta_h
$$
is a propositional tautology. Thus, by Necessitation inference rule,
$$
\vdash \Box_e\left(\delta\to\delta_e\vee \bigvee_{h\in C^u_{\miniminus e}}\delta_h\vee \bigvee_{h\in C^{u'}_{\miniminus e}}\delta_h\right).
$$
By Distributivity axiom and Modus Ponens inference rule,
$$
\vdash \Box_e\delta\to\Box_e\left(\delta_e\vee \bigvee_{h\in C^u_{\miniminus e}}\delta_h\vee \bigvee_{h\in C^{u'}_{\miniminus e}}\delta_h\right).
$$
Using statement~(\ref{tuesday3}) and the laws of propositional logic,
$$
\vdash \Box_e\delta\to \delta_e\vee\Box_e\bigvee_{h\in C^u_{\miniminus e}}\delta_h\vee\Box_e\bigvee_{h\in C^{u'}_{\miniminus e}}\delta_h.
$$
Recall that $\Box_e\delta\in M$ and $\delta_e\notin M$. Thus, $\Box_e\bigvee_{h\in C^u_{\miniminus e}}\delta_h\in M$ or $\Box_e\bigvee_{h\in C^{u'}_{\miniminus e}}\delta_h\in M$, due to the maximality and the consistency of set $M$. In either case, the required follows from Lemma~\ref{path exists}.
\end{proof}

In general, the completeness of a modal logic is often proven through a construction that converts a maximal consistent set of formulas into a world of a ``canonical" model for this set of formulas. In our case, the canonical model is represented by protocol $\mathcal{P}^\delta_E$. Instead of a Kripke world, we construct a special run of this protocol. The construction is done recursively for an arbitrary $\mathcal{P}^\delta_F$ in the theorem below. Informally, in term of the network flow model, the theorem states that for any maximal consistent set of formulas $X$ there is a network flow on the graph that satisfies this set of formulas.

\begin{theorem}\label{run exists}
For every $\delta\in\Delta(\Sigma)$ every $F\subseteq E$ and every maximal consistent set $M$ there is a run $r=\langle X^e, \{f^e_{u} \}_{u\in Inc(e)}\rangle_{e\in E}$ of protocol ${\cal P}^\delta_F$ such that  for each $e\in E$, we have $X^e =M \cap \Phi(\Sigma,\{e\})$.
\end{theorem}
\begin{proof} 
We prove the theorem by induction on the size of set $F$. 

If $F=\varnothing$, for each $e\in E$ and each $u\in Inc(e)$, let 
$$
f^e_u =
\begin{cases}
1, &\mbox{if $\delta_e\in X^e$},\\
0, &\mbox{otherwise}. 
\end{cases}
$$

\begin{claim}
Tuple $\langle X^e, \{f^e_{u} \}_{u\in Inc(e)}\rangle_{e\in E}$ is a run of protocol ${\cal P}^\delta_{\varnothing}$.
\end{claim}
\begin{proof}
The claim immediately follows from Definition~\ref{values} and the definition of local conditions of protocol ${\cal P}^\delta_\varnothing$ on page~\pageref{local conditions}.
\end{proof}

Next, assume that $F = F' \cup \{h\}$. By the induction hypothesis, there is a run $r=\langle X^e, \{f^e_{u} \}_{u\in Inc(e)}\rangle_{e\in E}$ of protocol ${\cal P}^\delta_{F'}$ such that $X^e =M \cap \Phi(\Sigma,\{e\})$ for each $e\in E$. If $\Box_h\delta\notin X^h$ or $\delta_h\in X^h$, then, by Definition~\ref{values}, run $r$ is a run of protocol ${\cal P}^\delta_{F}$. Suppose now that $\Box_h\delta\in X^h$ and $\delta_h\notin X^h$. Let $\lambda$ be any positive real number such that
$$
\lambda > |f^e_u|
$$
for each $e\in E$ and each $u\in Inc(e)$. By the assumption $\Box_h\delta\in X^h$ and Lemma~\ref{path exists 2}, there is a path $e_0,v_1,e_1,v_2,e_2,\dots,v_k,e_k$ in $\Gamma_M$ such that $e_0=h$. Let $v_0$ be the end of edge $h$ different from $v_1$ and let $v_{k+1}$ be the end of edge $e_k$ different from $v_k$. We next define a tuple $\widehat{r}=\langle X^e, \{\widehat{f}^e_{u} \}_{u\in Inc(e)}\rangle_{e\in E}$, for which we consider two cases, see Figures~\ref{lambda bridge figure} and \ref{lambda nobridge figure}:

\begin{figure}[ht]
\begin{center}
\vspace{3mm}
\scalebox{.5}{\includegraphics{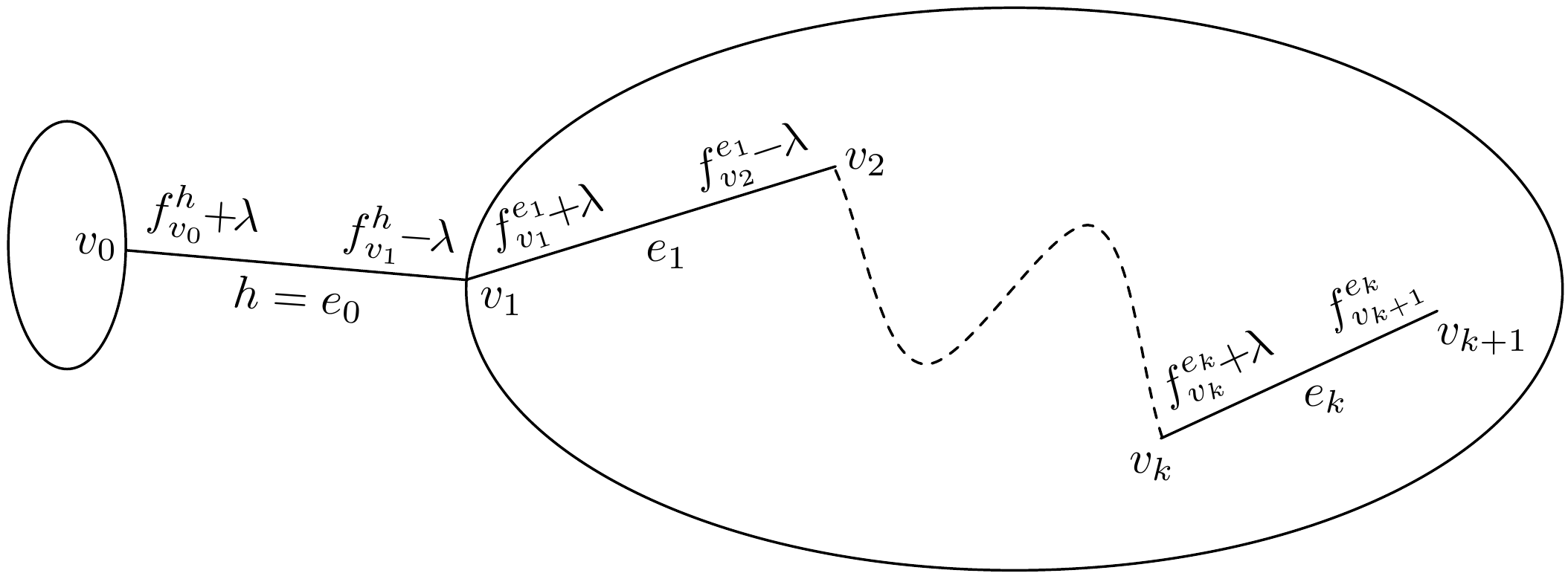}}
\vspace{0mm}
\caption{Definition of $\widehat{f}^e_u$ if $h\in \cal B$.}\label{lambda bridge figure}
\vspace{0cm}
\end{center}
\vspace{0cm}
\end{figure}

\noindent{Case I}: If $h\in \cal B$, then for each $e\in E$ and each $u\in Inc(e)$,
$$
\widehat{f}^e_u =
\begin{cases}
f^{e}_{u}+\lambda, & \mbox{where $e=e_i$, $u=v_i$, and $0\le i\le k$,}\\
f^{e}_{u}-\lambda, & \mbox{where $e=e_i$, $u=v_{i+1}$, and $0\le i<k$,}\\
f^e_u, & \mbox{otherwise.}
\end{cases}
$$

\begin{figure}[ht]
\begin{center}
\vspace{3mm}
\scalebox{.5}{\includegraphics{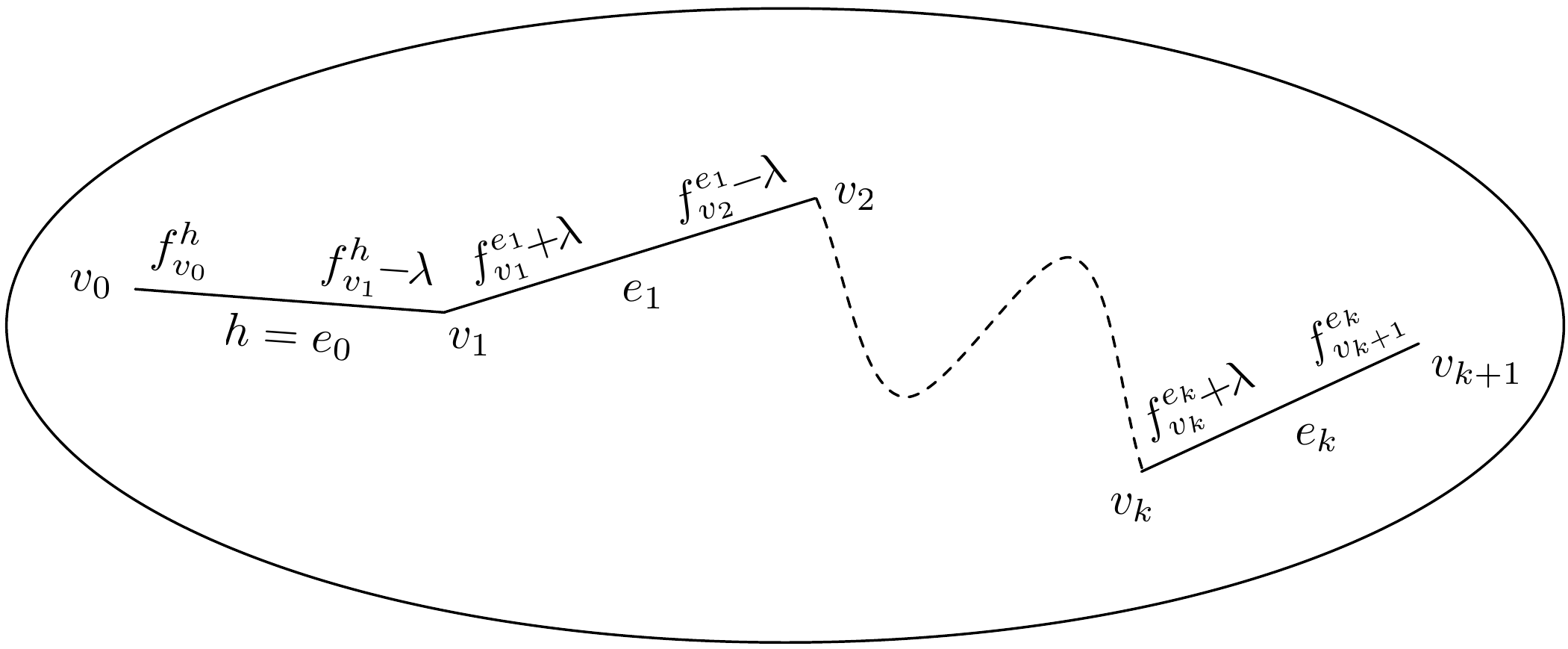}}
\vspace{0mm}
\caption{Definition of $\widehat{f}^e_u$ if $h\in E\setminus  \cal B$.}\label{lambda nobridge figure}
\vspace{0cm}
\end{center}
\vspace{0cm}
\end{figure}

\noindent{Case II}: If $h\in E\setminus \cal B$,  then for each $e\in E$ and each $u\in Inc(e)$,
$$
\widehat{f}^e_u =
\begin{cases}
f^{e}_{u}+\lambda, & \mbox{where $e=e_i$, $u=v_i$, and $0< i\le k$,}\\
f^{e}_{u}-\lambda, & \mbox{where $e=e_i$, $u=v_{i+1}$, and $0\le i<k$,}\\
f^e_u, & \mbox{otherwise.}
\end{cases}
$$
This defines tuple $\widehat{r}$.

\begin{claim}\label{hat removal}
$\widehat{f}^e_u+\widehat{f}^e_{u'}=f^e_u+f^e_{u'}$, for each $e\in Edge(u,u')\in E\setminus\{e_0,e_k\}$.
\end{claim}
\begin{proof}
If $e=e_i$ for some $0<i<k$, then 
$$\widehat{f}^e_u+\widehat{f}^e_{u'}=\widehat{f}^e_{v_i}+\widehat{f}^e_{v_{i+1}}=
{f}^e_{v_i}+\lambda+{f}^e_{v_{i+1}}-\lambda={f}^e_{v_i}+{f}^e_{v_{i+1}}
=f^e_u+f^e_{u'}.$$
Otherwise, $\widehat{f}^e_u={f}^e_{u}$ and $\widehat{f}^e_{u'}={f}^e_{u'}$. Thus, $\widehat{f}^e_u+\widehat{f}^e_{u'}=f^e_u+f^e_{u'}$.
\end{proof}

\begin{claim}
Tuple $\widehat{r}$ is a run of protocol ${\cal P}^\delta_{F}$.
\end{claim}
\begin{proof}
We need to verify that the tuple $\widehat{r}$ satisfies the conditions of Definition~\ref{values} and the local conditions of the run ${\cal P}^\delta_{F}$ on page~\pageref{local conditions}. Below by $v_{k+1}$  we denote the end of edge $e_k$ different from vertex $v_k$. We start with conditions of Definition~\ref{values}.
\begin{itemize}
\item[{\sf 1(c)}] Due to Claim~\ref{hat removal} and the assumption that $r$ is a run of protocol ${\cal P}^\delta_{F'}$, we only need to verify condition 1(c) for edges $e_0$ and $e_k$. 

We first verify this condition for edge $e_0$. Note that $e_0=h$. Thus, $\delta_{e_0}\notin X^{e_0}$ due to our assumption. Hence, $f^{e_0}_{v_0}+f^{e_0}_{v_1}\le 0$, because run $r$ satisfies condition 1(c) of Definition~\ref{values}.

If $e_0\in \cal B$, then  
$$
\widehat{f}^{e_0}_{v_0}+\widehat{f}^{e_0}_{v_1}=
{f}^{e_0}_{v_0}+\lambda+{f}^{e_0}_{v_1}-\lambda={f}^{e_0}_{v_0}+{f}^{e_0}_{v_1}\le 0.
$$
If $e_0\notin \cal B$, then, since $\lambda>0$, 
$$
\widehat{f}^{e_0}_{v_0}+\widehat{f}^{e_0}_{v_1}=
{f}^{e_0}_{v_0}+{f}^{e_0}_{v_1}-\lambda<{f}^{e_0}_{v_0}+{f}^{e_0}_{v_1}\le 0.
$$
In either case, we have $\delta_{e_0}\notin X^{e_0}$ and $\widehat{f}^{e_0}_{u}+\widehat{f}^{e_0}_{u'}\le 0$. Thus, condition 1(c) is satisfied.

Next, we verify this condition for the edge $e_k$.  Note that $\delta_{e_k}\in X^{e_k}$, by Definition~\ref{gamma}. Thus, we only need to show that $\widehat{f}^{e_k}_{v_k}+\widehat{f}^{e_k}_{v_{k+1}}>0$. Indeed, ${f}^{e_k}_{v_k}+{f}^{e_k}_{v_{k+1}}>0$ because run $r$ satisfies condition 1(c). Thus, since $\lambda > 0$,
$$
\widehat{f}^{e_k}_{v_k}+\widehat{f}^{e_k}_{v_{k+1}}={f}^{e_k}_{v_k}+ \lambda +{f}^{e_k}_{v_{k+1}} > {f}^{e_k}_{v_k}+{f}^{e_k}_{v_{k+1}} > 0.
$$
\item[{\sf 2(a)}] Due to Claim~\ref{hat removal} and the assumption that $r$ is a run of protocol ${\cal P}^\delta_{F'}$, we again only need to verify condition 2(a) for edges $e_0$ and $e_k$. 

We first verify this condition for edge $e_0$. Note that $\delta_{e_0}\notin X^{e_0}$ by condition 2 of Definition~\ref{gamma}. Since run $r$ satisfies the condition 2(c) of Definition~\ref{values}, we have ${f}^{e_0}_{v_0}+{f}^{e_0}_{v_1}=0$. Hence,
$$
\widehat{f}^{e_0}_{v_0}+\widehat{f}^{e_0}_{v_1}=
{f}^{e_0}_{v_0}+\lambda+{f}^{e_0}_{v_1}-\lambda={f}^{e_0}_{v_0}+{f}^{e_0}_{v_1}= 0.
$$

For edge $e_k$ this condition is vacuously true because $\delta_{e_k}\in X^{e_k}$ due to condition 4 of Definition~\ref{gamma}.

\item[{\sf 2(b)}] By the definition of $\widehat{r}$, for each edge $b \in {\cal B} \setminus \{e_0,\dots,e_k\}$, and each vertex $u\in Inc(b)$, we have $\widehat{f}^b_u=f^b_u$. Thus, $\widehat{r}$ on any such edge satisfies condition 2(b) of Definition~\ref{values} because run $r$ does. 

We next show that condition 2(b) is satisfied for each $e_i$ such that $e_i\in \cal B$ and $0\le i\le k$. Indeed, consider any $u\in Inc(e_i)$ and suppose that $\widehat{f}^{e_i}_u<0$. 

If $u=v_i$, then, since $\lambda>0$,
$$
f^{e_i}_u=f^{e_i}_{v_i}=\widehat{f}^{e_i}_{v_i}-\lambda<\widehat{f}^{e_i}_u<0.
$$
Thus, $\Box_{e_i}\bigvee_{e\in C^u_{\miniminus e_i}}\delta_e\in X^{e_i}$ because run $r$ satisfies condition 2(b) of Definition~\ref{values}.

If $u=v_{i+1}$ and $i<k$, then condition 2(b) is satisfied due to condition 3 of Definition~\ref{gamma}.

Finally, if $i=k$ and $u=v_{k+1}$, then $\widehat{f}^{e_i}_u={f}^{e_i}_u$ by the definition of $\widehat{r}$. Thus, condition 2(b) is satisfied by run $\widehat{r}$ because it is satisfied by run $r$.
\item[{\sf 2(c)}] By the definition of $\widehat{r}$, for each edge $b \in {\cal B} \setminus \{e_0,\dots,e_k\}$, and each vertex $u\in Inc(b)$, we have $\widehat{f}^b_u=f^b_u$. Thus, $\widehat{r}$ on any such edge satisfies condition 2(c) of Definition~\ref{values} because run $r$ does. 

We will next show that condition 2(c) is satisfied for each $e_i$ such that $e_i\in \cal B$ and $0\le i< k$. Indeed, note that $\lambda>|f^{e_i}_{v_{i+1}}|$ due to the choice of $\lambda$. Thus
$$
\widehat{f}^{e_i}_{v_{i+1}}={f}^{e_i}_{v_{i+1}}-\lambda<0. 
$$
Finally, note that when $i=k$, we have $\delta_{e_k}\in X^{e_k}$. Therefore, condition 2(c) is vacuously true. 

\item[{\sf 3(a)}] Due to Claim~\ref{hat removal} and the assumption that $r$ is a run of protocol ${\cal P}^\delta_{F'}$, we again only need to verify condition 3(a) for edges $e_0$ and $e_k$. 

Note that $\Box_{h}\delta\in X^{h}$ by our assumption. Recall that $e_0=h$. Thus, $\Box_{e_0}\delta\in X^{e_0}$. Therefore, condition 3(a) is satisfied for edge $e_0$.

By condition 4 of Definition~\ref{gamma}, $\delta_{e_k}\in X^{e_k}$. Thus, as we have shown in the case 1(c) above, $\widehat{f}^{e_k}_{v_k}+\widehat{f}^{e_k}_{v_{k+1}}>0$. Therefore, condition 3(a) is vacuously true for edge $e_k$.

\item[{\sf 3(b)}] Due to Claim~\ref{hat removal} and the assumption that $r$ is a run of protocol ${\cal P}^\delta_{F'}$, we again only need to verify condition 3(b) for edges $e_0$ and $e_k$. 

Note that $\delta_h\notin X^{h}$ by our assumption. Recall that $e_0=h$. Thus, $\delta_{e_0}\notin X^{e_0}$. Since $r$ is a run of protocol ${\cal P}^{\delta}_{F'}$, by condition 1(c) of Definition~\ref{values}, we have ${f}^{e_0}_{v_0}+{f}^{e_0}_{v_1}\le 0$. Hence, due to $\lambda>0$,
$$
\widehat{f}^{e_0}_{v_0}+\widehat{f}^{e_0}_{v_1}={f}^{e_0}_{v_0}+{f}^{e_0}_{v_1}-\lambda \le 0 -\lambda < 0.
$$
Therefore, condition 3(b) is satisfied for edge $e_0$.
By condition 4 of Definition~\ref{gamma}, $\delta_{e_k}\in X^{e_k}$. Thus, condition 3(b) is vacuously true for edge $e_k$.
\end{itemize}
To show that local conditions (see page~\pageref{local conditions}) are satisfied at any vertex $u\in V$, it is sufficient to show that 
$$
\sum_{e\in Inc(u)}\widehat{f}^{e}_{u}
\ge \sum_{e\in Inc(u)}{f}^{e}_u.
$$

Consider first the case when $u=v_0$ and $e_0\in \cal B$. Since it has been assumed (see page~\pageref{simple path}) that vertices along any path do not repeat and because $\lambda>0$,
\begin{eqnarray*}
\sum_{e\in Inc(v_0)}\widehat{f}^{e}_{v_0} = 
\widehat{f}^{e_0}_{v_0} + \sum_{e\in Inc(v_0)\setminus\{e_0\}}\widehat{f}^{e}_{v_0}=
{f}^{e_0}_{v_0} +\lambda +
\sum_{e\in Inc(v_0)\setminus\{e_0\}}{f}^{e}_{v_0} \\
=\sum_{e\in Inc(v_0)}{f}^{e}_{v_0} +\lambda > \sum_{e\in Inc(v_0)}{f}^{e}_{v_0}.
\end{eqnarray*}

Next, consider the case when vertex $u=v_i$ for some $0<i\le k$. Then,
\begin{eqnarray*}
\sum_{e\in Inc(v_i)}\widehat{f}^{e}_{v_i} = \widehat{f}^{e_{i-1}}_{v_i}+\widehat{f}^{e_i}_{v_i}+ \sum_{e\in Inc(v_i)\setminus\{e_{i-1},e_i\}}\widehat{f}^{e}_{v_i}\\=
{f}^{e_{i-1}}_{v_i}-\lambda +{f}^{e_i}_{v_i}+\lambda +
\sum_{e\in Inc(v_i)\setminus\{e_{i-1},e_i\}}{f}^{e}_{v_i} = 
\sum_{e\in Inc(v_i)}{f}^{e}_{v_i}.
\end{eqnarray*}

Otherwise, the sum $\sum_{e\in Inc(u)}\widehat{f}^{e}_{u}$
and the sum $\sum_{e\in Inc(u)}{f}^{e}_u$ are equal because they consist of equal terms.
\end{proof}
This concludes the proof of Theorem~\ref{run exists}.
\end{proof}

 The previous theorem constructs a run (``epistemic world") that matches a maximal consistent set $M$ on all edges. The next theorem enhances the claim of the previous theorem by adding an additional condition on the run being constructed. Namely, if $h$ is a given edge of the graph and $r$ is a given run of the protocol, then the desired run $\widehat{r}$ can be constructed not only to match set $M$ on all edges, but also to satisfy the equation $\widehat{r}=_h r$. The theorem assumes, of course, that run $r$ itself matches set $M$ on edge $h$. In terms of the network flow model, the theorem states that if there is a network flow that satisfies local properties $M\cap \Phi(\Sigma,\{h\})$ at a given edge $h$, then this network flow can be modified to match properties in $M$ globally (on all edges of the graph). The proof of the theorem below explains how the water can be re-routed through the graph to achieve the desired outcome.

\begin{theorem}\label{run exists 2}
For each $h\in E$, each run $r=\langle X^e, \{f^e_{u} \}_{u\in Inc(e)}\rangle_{e\in E}$ of protocol ${\cal P}^\delta_E$, and each maximal consistent set $M$ such that $X^h=M\cap \Phi(\Sigma,\{h\})$, there is a run 
$$\widehat{r}=\langle \widehat{X}^e, \{\widehat{f}^e_{u} \}_{u\in Inc(e)} \rangle_{e\in E}$$ of protocol ${\cal P}^\delta_E$ such that
\begin{enumerate}
\item $\widehat{X}^e=M\cap \Phi(\Sigma,\{e\})$ for each $e\in E$,
\item $\widehat{r}=_h r$.
\end{enumerate}
\end{theorem}
\begin{proof}
By Theorem~\ref{run exists}, there is a run $r'=\langle Y^e, \{\ell^e_{u} \}_{u\in Inc(e)} \rangle_{e\in E}$ of protocol ${\cal P}^\delta_E$ such that $Y^e=M\cap \Phi(\Sigma,\{e\})$ for each $e\in E$. We will show how this run can be modified to obtain the desired run $\widehat{r}$, by considering several possible cases.

\noindent{\bf  Case I}: if $\delta_h\in M$, then define $\widehat r$ to be the tuple $\langle Y^e, \{\widehat{f}^e_{u} \}_{u\in Inc(e)} \rangle_{e\in E}$, where
$$
\widehat{f}^e_u=
\begin{cases}
f^h_{u}, &\mbox{ if $e=h$,}\\
\ell^e_u, & \mbox{ otherwise.} 
\end{cases}
$$

\begin{claim}
$\widehat r$ is a run of protocol ${\cal P}^\delta_E$ and $\widehat{r}=_h r$.
\end{claim}
\begin{proof}
We need to verify that tuple $\widehat{r}$ satisfies conditions of Definition~\ref{values} and the local conditions of protocol ${\cal P}^\delta_{E}$ on page~\pageref{local conditions}. 

We start with the conditions of Definition~\ref{values} for an arbitrary edge $e\in E$.
If $e=h$, then $\widehat{r}=_e r$, and thus tuple $\widehat{r}$ satisfies the conditions of Definition~\ref{values} on edge $e$ because run $r$ does. Similarly, if $e\neq h$, then $\widehat{r}=_e r'$, and thus tuple $\widehat{r}$ satisfies the conditions of Definition~\ref{values} on edge $e$ because run $r'$ does.

We now show that tuple $\widehat{r}$ vacuously satisfies local conditions of protocol ${\cal P}^\delta_{E}$ at any vertex $v\in V$. If $v\notin Inc(h)$, then $\widehat{r}=_e r'$ for each $e\in Inc(v)$. Thus, tuple $\widehat{r}$ satisfies local conditions of protocol ${\cal P}^\delta_{E}$ because run $r'$ does. If $v\in Inc(h)$, then tuple $\widehat{r}$ vacuously satisfies local conditions of protocol ${\cal P}^\delta_{E}$ because $\delta_h\in M$.

The condition $\widehat{r}=_h r$ is satisfied because (i) $Y^h=M\cap \Phi(\Sigma,\{h\})=X^h$ and (ii)
$\widehat{f}^h_u=f^h_u$ for each $u\in Inc(h)$.
\end{proof}

\noindent{\bf  Case II}: if $\delta_h\notin M$ and $h\in E\setminus \cal B$. Let $h\in Edge(v_0,v_1)$. Since $h\notin \cal B$, there is a circular path 
$h=e_0,v_1,e_1,v_2,\dots,v_{k-1},e_{k-1},v_k,e_k=h$. By Definition~\ref{bridge}, $e_i\notin \cal B$ for each $0\le i< k$. We will now further split this case into two subcases:

\vspace{2mm}
\noindent{\bf Subcase IIa}: If $\Box_h\delta\notin M$, then define  $\widehat r$ to be tuple $\langle Y^e, \{\widehat{f}^e_{u} \}_{u\in Inc(e)} \rangle_{e\in E}$, see Figure~\ref{subcase-II-a figure}, where
$$
\widehat{f}^e_u=
\begin{cases}
\ell^e_{u}+f^h_{v_0}-\ell^h_{v_0}, &\mbox{ if $e=e_i$, $u=v_i$, and $0\le i<k$,}\\
\ell^e_{u}+f^h_{v_1}-\ell^h_{v_1}, &\mbox{ if $e=e_i$, $u=v_{i+1}$, and $0\le i<k$,}\\
\ell^e_u, & \mbox{ otherwise.} 
\end{cases}
$$

\begin{figure}[ht]
\begin{center}
\vspace{3mm}
\scalebox{.5}{\includegraphics{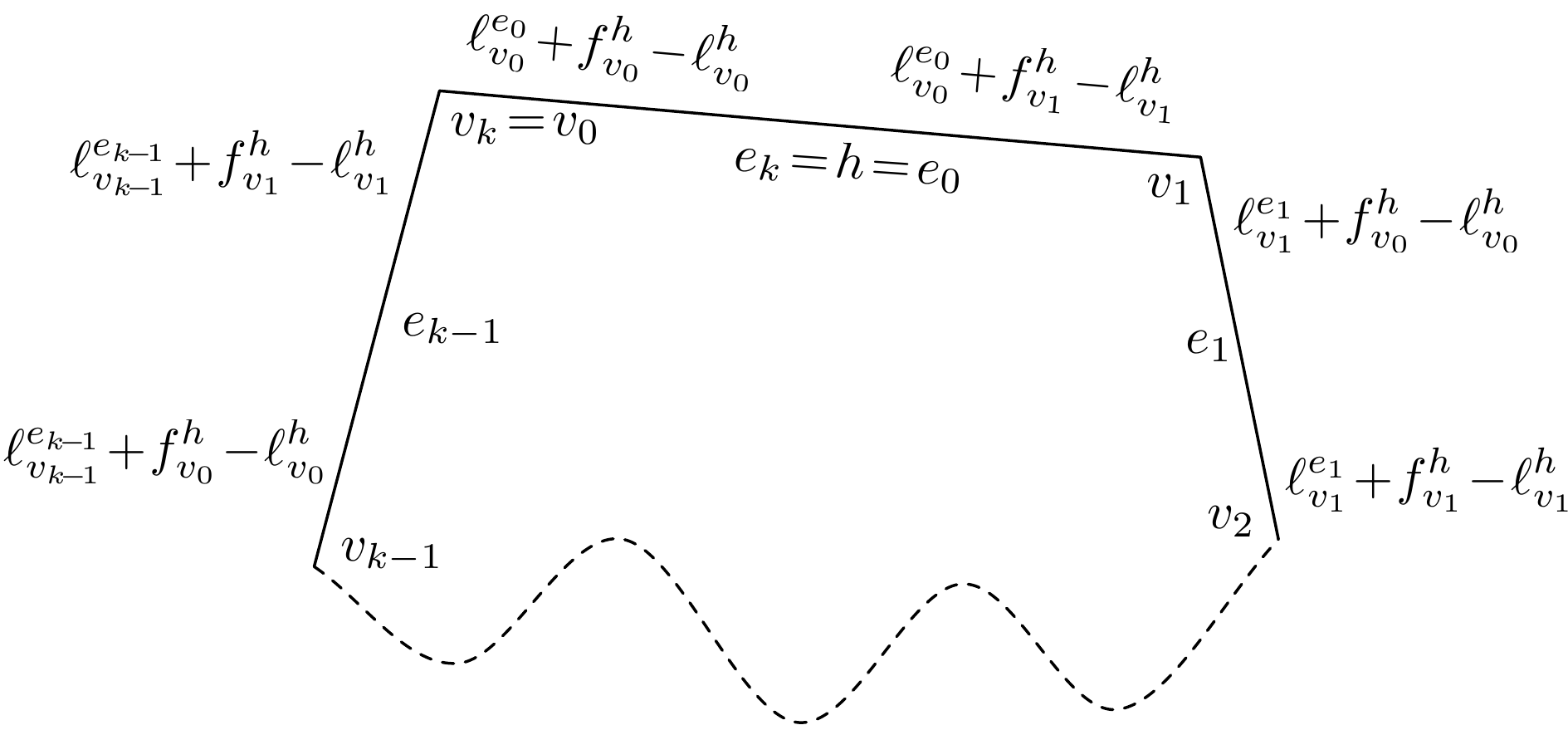}}
\vspace{0mm}
\caption{Subcase IIa.}\label{subcase-II-a figure}
\vspace{0cm}
\end{center}
\vspace{0cm}
\end{figure}

\begin{claim}\label{two equalities}
$\widehat{f}^{e_i}_{v_i}+\widehat{f}^{e_i}_{v_{i+1}}=\ell^{e_i}_{v_i}+\ell^{e_i}_{v_{i+1}}$ and
$\widehat{f}^{e_{i}}_{v_{i+1}}+\widehat{f}^{e_{i+1}}_{v_{i+1}}=\ell^{e_{i}}_{v_{i+1}}+\ell^{e_{i+1}}_{v_{i+1}}$, for each $0\le i<k$.
\end{claim}
\begin{proof}
By condition 1(c) of Definition~\ref{values}, the assumption $\delta_h\notin M$ implies that $f^h_{v_0}+f^h_{v_1}\le 0$ and $\ell^h_{v_0}+\ell^h_{v_1}\le 0$. By condition 3(a) of the same definition, the assumption $\Box_h\delta\notin M$ implies that $f^h_{v_0}+f^h_{v_1}\ge 0$ and $\ell^h_{v_0}+\ell^h_{v_1}\ge 0$. Thus, $f^h_{v_0}+f^h_{v_1}= 0$ and $\ell^h_{v_0}+\ell^h_{v_1}= 0$. Therefore,
\begin{eqnarray*}
\widehat{f}^{e_i}_{v_i}+\widehat{f}^{e_i}_{v_{i+1}}=\ell^{e_i}_{v_i}+f^h_{v_0}-\ell^h_{v_0}+\ell^{e_i}_{v_{i+1}}+f^h_{v_1}-\ell^h_{v_1}=
\ell^{e_i}_{v_i}+\ell^e_{v_{i+1}}+\\(f^h_{v_0}+f^h_{v_1})-(\ell^h_{v_0}+\ell^h_{v_1})=
\ell^{e_i}_{v_i}+\ell^e_{v_{i+1}}+0 - 0
=\ell^{e_i}_{v_i}+\ell^{e_i}_{v_{i+1}},
\end{eqnarray*}
and
\begin{eqnarray*}
\widehat{f}^{e_{i}}_{v_{i+1}}+\widehat{f}^{e_{i+1}}_{v_{i+1}}=
\ell^{e_i}_{v_{i+1}}+f^h_{v_1}-\ell^h_{v_1}+
\ell^{e_{i+1}}_{v_{i+1}}+f^h_{v_0}-\ell^h_{v_0}=
\ell^{e_i}_{v_{i+1}}+\ell^{e_{i+1}}_{v_{i+1}}+\\
(f^h_{v_0}+f^h_{v_1})-(\ell^h_{v_0}+\ell^h_{v_1})=
\ell^{e_i}_{v_{i+1}}+\ell^{e_{i+1}}_{v_{i+1}} + 0 - 0=
\ell^{e_i}_{v_{i+1}}+\ell^{e_{i+1}}_{v_{i+1}}.
\end{eqnarray*}
\end{proof}

\begin{claim}
 $\widehat r$ is a run of protocol ${\cal P}^\delta_E$  and $\widehat{r}=_h r$.
\end{claim}
\begin{proof}
We need to verify that the tuple $\widehat{r}$ satisfies the conditions of Definition~\ref{values} and the local conditions of protocol ${\cal P}^\delta_{E}$ on page~\pageref{local conditions}. 

We start with the conditions of Definition~\ref{values} for an arbitrary edge $e\in E$.
If $e=e_i$ for some $0\le i< k$, then, due to the path being circular, $e\notin \cal B$. Thus, all applicable conditions from Definition~\ref{values} are satisfied for tuple $\widehat{r}$ because they are satisfied for run $r'$ and due to the equality 
$\widehat{f}^{e_i}_{v_i}+\widehat{f}^{e_i}_{v_{i+1}}=\ell^{e_i}_{v_i}+\ell^{e_i}_{v_{i+1}}$
established in Claim~\ref{two equalities}.  If $e\neq e_i$ for all $0\le i< k$, then the required is true because $\widehat{r}=_e r'$.

We now show that tuple $\widehat{r}$ satisfies local conditions of protocol ${\cal P}^\delta_{E}$ at any vertex $v\in V$. If $v=v_{i+1}$ for some $0\le i< k$, then $\widehat{f}^{e_{i}}_{v_{i+1}}+\widehat{f}^{e_{i+1}}_{v_{i+1}}=\ell^{e_{i}}_{v_{i+1}}+\ell^{e_{i+1}}_{v_{i+1}}$ by Claim~\ref{two equalities}. Thus, $\sum_{e\in Inc(v_{i+1})}\widehat{f}^{e}_{v_{i+1}}=\sum_{e\in Inc(v_{i+1})}\ell^{e}_{v_{i+1}}$. If $v\neq v_{i+1}$ for all $0\le i< k$, then  $\widehat{r}=_e r'$ for all $e\in Inc(v)$. In either of these two cases, tuple $\widehat{r}$ satisfies the local conditions  of protocol ${\cal P}^\delta_{E}$ at vertex $v\in V$ because run $r'$ satisfies these conditions.

Condition $\widehat{r}=_h r$ is satisfied because (i) $Y^h=M\cap \Phi(\Sigma,\{h\})=X^h$, (ii)
$\widehat{f}^h_{v_0}=\ell^h_{v_0}+f^h_{v_0}-\ell^h_{v_0}=f^h_{v_0}$, and (iii)
$\widehat{f}^h_{v_1}=\ell^h_{v_1}+f^h_{v_1}-\ell^h_{v_1}=f^h_{v_1}$.
\end{proof}

\noindent{\bf Subcase IIb}: If $\Box_h\delta\in M$, then $f^h_{v_0}+f^h_{v_1}<0$ and $\ell^h_{v_0}+\ell^h_{v_1}<0$ due to condition 3(b) of Definition~\ref{values}. Let 
$\lambda=(f^h_{v_0}+f^h_{v_1})/(\ell^h_{v_0}+\ell^h_{v_1})$. Note that $\lambda>0$. Define $\widehat r$ to be the tuple $\langle Y^e, \{\widehat{f}^e_{u} \}_{u\in Inc(e)} \rangle_{e\in E}$, see Figure~\ref{subcase-II-b figure}, where
$$
\widehat{f}^e_u=
\begin{cases}
\lambda(\ell^e_u-\ell^h_{v_0})+f^h_{v_0}, &\mbox{ if $e=e_i$, $u=v_i$, and $0\le i<k$,}\\
\lambda(\ell^e_u+\ell^h_{v_0})-f^h_{v_0}, &\mbox{ if $e=e_i$, $u=v_{i+1}$, and $0\le i<k$,}\\
\lambda\ell^e_u, & \mbox{ otherwise.} 
\end{cases}
$$

\begin{figure}[ht]
\begin{center}
\vspace{3mm}
\scalebox{.5}{\includegraphics{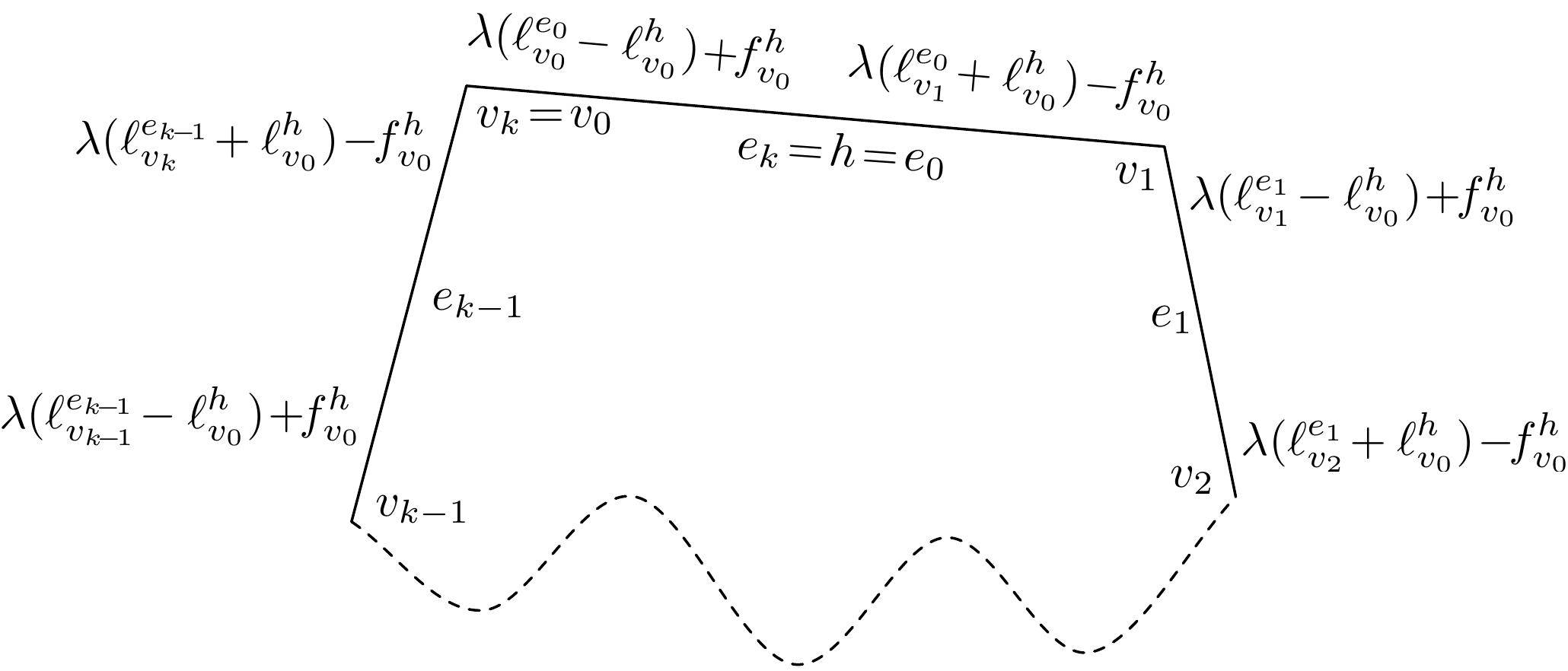}}
\vspace{0mm}
\caption{Subcase IIb.}\label{subcase-II-b figure}
\vspace{0cm}
\end{center}
\vspace{0cm}
\end{figure}

\begin{claim}\label{two equalities 2}
$\widehat{f}^{e_i}_{v_i}+\widehat{f}^{e_i}_{v_{i+1}}=\lambda(\ell^{e_i}_{v_i}+\ell^{e_i}_{v_{i+1}})$ and
$\widehat{f}^{e_{i}}_{v_{i+1}}+\widehat{f}^{e_{i+1}}_{v_{i+1}}=\lambda(\ell^{e_{i}}_{v_{i+1}}+\ell^{e_{i+1}}_{v_{i+1}})$, for each $0\le i<k$.
\end{claim}
\begin{proof}
\begin{eqnarray*}
\widehat{f}^{e_i}_{v_i} + \widehat{f}^{e_i}_{v_{i+1}} = 
\lambda(\ell^{e_i}_{v_i}-\ell^h_{v_0})+f^h_{v_0} + \lambda(\ell^{e_i}_{v_{i+1}}+\ell^h_{v_0})-f^h_{v_0}
= \lambda(\ell^{e_i}_{v_i}+\ell^{e_i}_{v_{i+1}}).
\end{eqnarray*}
Similarly,
$$
\widehat{f}^{e_{i}}_{v_{i+1}}+\widehat{f}^{e_{i+1}}_{v_{i+1}}=
\lambda(\ell^{e_i}_{v_{i+1}}+\ell^h_{v_0})-f^h_{v_0} + \lambda(\ell^{e_{i+1}}_{v_{i+1}}-\ell^h_{v_0})+f^h_{v_0}=\lambda(\ell^{e_i}_{v_{i+1}}+\ell^{e_{i+1}}_{v_{i+1}}).
$$
\end{proof}

\begin{claim}
$\widehat r$ is a run of protocol ${\cal P}^\delta_E$  and $\widehat{r}=_h r$.
\end{claim}
\begin{proof}
We need to verify that tuple $\widehat{r}$ satisfies the conditions of Definition~\ref{values} and the local conditions of protocol ${\cal P}^\delta_{E}$ on page~\pageref{local conditions}. 

We start with the conditions of Definition~\ref{values} for an arbitrary edge $e\in E$.
If $e=e_i$ for some $0\le i< k$, then $e\notin \cal B$ since the path is circular. Thus, all applicable conditions from Definition~\ref{values} are satisfied for tuple $\widehat{r}$ because they are satisfied for run $r'$ and due to $\lambda>0$ and the equality 
$\widehat{f}^{e_i}_{v_i}+\widehat{f}^{e_i}_{v_{i+1}}=\lambda(\ell^{e_i}_{v_i}+\ell^{e_i}_{v_{i+1}})$
established in Claim~\ref{two equalities 2}.  If $e\neq e_i$ for all $0\le i< k$, then the required is true because run $r'$ satisfies the conditions from Definition~\ref{values} and $\widehat{f}^e_{u}=_e \lambda \ell^e_u$ for each $u\in Inc(e)$, where $\lambda>0$. 

We now show that tuple $\widehat{r}$ satisfies the local conditions of protocol ${\cal P}^\delta_{E}$ at any vertex $v\in V$. If $v=v_{i+1}$ for some $0\le i< k$, then $\widehat{f}^{e_{i}}_{v_{i+1}}+\widehat{f}^{e_{i+1}}_{v_{i+1}}=\lambda(\ell^{e_{i}}_{v_{i+1}}+\ell^{e_{i+1}}_{v_{i+1}})$ by Claim~\ref{two equalities 2}. Thus, $\sum_{e\in Inc(v_{i+1})}\widehat{f}^{e}_{v_{i+1}}=\lambda\sum_{e\in Inc(v_{i+1})}\ell^{e}_{v_{i+1}}$. If $v\neq v_{i+1}$ for all $0\le i< k$, then  $\widehat{f}^e_v=\lambda \ell^e_v$ for all $e\in Inc(v)$. In either of these two cases, tuple $\widehat{r}$ satisfies the local conditions  of protocol ${\cal P}^\delta_{E}$ at vertex $v\in V$ because run $r'$ satisfies these conditions and $\lambda>0$.

The condition $\widehat{r}=_h r$ is satisfied because $Y^h=M\cap \Phi(\Sigma,\{h\})=X^h$,
\begin{eqnarray*}
\widehat{f}^h_{v_0}=\lambda(\ell^h_{v_0}-\ell^h_{v_0})+f^h_{v_0}=0+f^h_{v_0}=f^h_{v_0},
\end{eqnarray*}
and
\begin{eqnarray*}
\widehat{f}^h_{v_1}=\lambda(\ell^h_{v_1}+\ell^h_{v_0})-f^h_{v_0}
=\dfrac{f^h_{v_0}+f^h_{v_1}}{\ell^h_{v_0}+\ell^h_{v_1}}(\ell^h_{v_1}+\ell^h_{v_0})-f^h_{v_0}
=f^h_{v_0}+f^h_{v_1}-f^h_{v_0}=f^h_{v_1}.
\end{eqnarray*}
\end{proof}

\noindent{\bf Case III}: If $\delta_h\notin M$ and $h\in \cal B$. Let $h\in Edge(v_0,v_1)$. There are three subcases:

\vspace{2mm}
\noindent{\bf Subcase IIIa}:
If $f^h_{v_1}\cdot\ell^h_{v_1}=0$, then $f^h_{v_1}=0$ or $\ell^h_{v_1}=0$. Hence, by Lemma~\ref{gateway zero}, $\Box_h\delta\notin X^h$. Thus, again by Lemma~\ref{gateway zero},  $f^h_{v_1}=0$, $f^h_{v_0}=0$, $\ell^h_{v_0}=0$, and $\ell^h_{v_1}=0$. Furthermore, $Y^h=M\cap \Phi(\Sigma,\{h\})=X^h$. Hence, $r=_h r'$. Let $\widehat{r}=r'$.

\vspace{2mm}
\noindent{\bf Subcase IIIb}: If 
$f^h_{v_1}\cdot\ell^h_{v_1}>0$, then define $\widehat r$ to be tuple $$\langle Y^e, \{(f^h_{v_1}/\ell^h_{v_1})\ell^e_u \}_{u\in Inc(e)} \rangle_{e\in E}.$$ 
By Corollary~\ref{scale} and the fact that $r'$ is a run of protocol ${\cal P}^\delta_E$, tuple $\widehat r$ is a run of protocol ${\cal P}^\delta_E$. Since $Y^h=M\cap \Phi(\Sigma,\{h\})=X^h$, to show that $\widehat{r}=_h r$, it is sufficient to show that $(f^h_{v_1}/\ell^h_{v_1})\ell^h_{v_1}=f^h_{v_1}$ and $(f^h_{v_1}/\ell^h_{v_1})\ell^h_{v_0}=f^h_{v_0}$. The former is an algebraic identity, the later follows from the equalities $f^h_{v_0}+f^h_{v_1}=0$ and $\ell^h_{v_0}+\ell^h_{v_1}=0$, which, in turn, follows from condition 2(a) of Definition~\ref{values}.

\vspace{2mm}
\noindent{\bf Subcase IIIc}: If $f^h_{v_1}\cdot\ell^h_{v_1}<0$, then $f^h_{v_1}\neq 0$. By Definition~\ref{values}, part 2(a), it follows that either $f^h_{v_1}<0$ or $f^h_{v_0}<0$. We consider the former case, the later one is similar. If $f^h_{v_1}<0$, then $\Box_h\bigvee_{e\in C_{\miniminus h}^{v_1}}\delta_e\in X^h$ by Definition~\ref{values}, part 2(b). Hence, $\Box_h\bigvee_{e\in C_{\miniminus h}^{v_1}}\delta_e\in M$. Thus, $\Box_h\bigvee_{e\in C_{\miniminus h}^{v_1}}\delta_e\in Y^h$. By Lemma~\ref{path exists}, there is a path $e_0,v_1,e_1,v_2,\dots,v_{k},e_k$ in $\Gamma_M$ such that $h=e_0$. Let $\lambda$ be any positive real number such that
$$
\lambda > |\ell^e_u|
$$
for each $e\in E$ and each $u\in Inc(e)$. Also, let $\mu=f^h_{v_0}/(\ell^h_{v_0}+\lambda)$. Recall that $f^h_{v_1}<0$. Thus, $f^h_{v_0}>0$ by condition 2(a) of Definition~\ref{values}. Additionally, note that $\lambda>|\ell^h_{v_0}|$. Thus, $\mu>0$.

Define $\widehat r$ to be tuple $\langle Y^e, \{\widehat{f}^e_{u} \}_{u\in Inc(e)} \rangle_{e\in E}$, see Figure~\ref{subcase-III-c figure}, where
\begin{equation}\label{f hat IIIc}
\widehat{f}^e_u=
\begin{cases}
\mu(\ell^e_{u}+\lambda), &\mbox{ if $e=e_i$, $u=v_i$, and $0\le i\le k$,}\\
\mu(\ell^e_{u}-\lambda), &\mbox{ if $e=e_i$, $u=v_{i+1}$, and $0\le i<k$,}\\
\mu\ell^e_u, & \mbox{ otherwise.} 
\end{cases}
\end{equation}

\begin{figure}[ht]
\begin{center}
\vspace{3mm}
\scalebox{.5}{\includegraphics{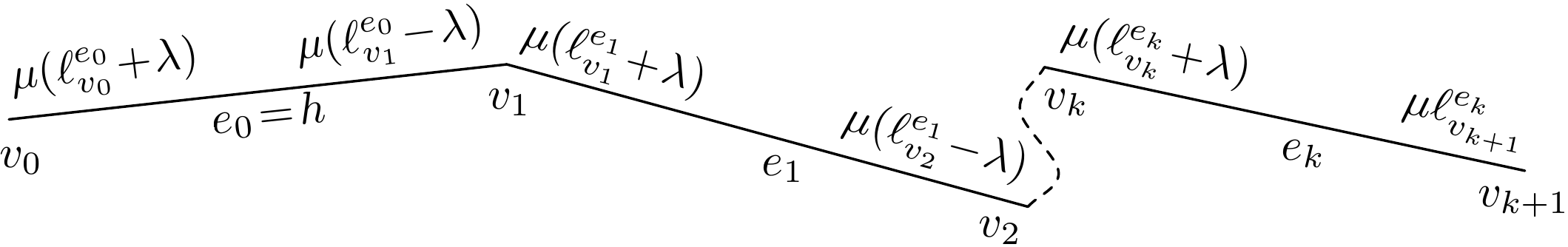}}
\vspace{0mm}
\caption{Subcase IIIc, the last vertex of the path, not named in the text, is denoted by $v_{k+1}$ on this figure.}\label{subcase-III-c figure}
\vspace{0cm}
\end{center}
\vspace{0cm}
\end{figure}

\begin{claim}\label{two equalities 3a}
$\widehat{f}^{e}_{u}+\widehat{f}^{e}_{u'}=\mu(\ell^{e}_{u}+\ell^{e}_{u'})$, for each edge $e\in Edge(u,u')\in E\setminus\{e_k\}$.
\end{claim}
\begin{proof}
If $e=e_i$ for some $0\le i< k$, then 
$$\widehat{f}^{e_i}_{v_i}+\widehat{f}^{e_i}_{v_{i+1}}=
\mu(\ell^{e_i}_{v_i}+\lambda) + \mu(\ell^{e_i}_{v_{i+1}}-\lambda)=
\mu(\ell^{e_i}_{v_i}+\ell^{e_i}_{v_{i+1}}).$$
If $e\neq e_i$ for all $0\le i\le k$, then
$\widehat{f}^{e}_{u}+\widehat{f}^{e}_{u}=\mu\ell^{e}_{u}+\mu\ell^{e}_{u'}=\mu(\ell^{e}_{u}+\ell^{e}_{u'})$.
\end{proof}

\begin{claim}\label{two equalities 3b}
$\sum_{e\in Inc(u)}\widehat{f}^e_u \ge \mu\sum_{e\in Inc(u)}\ell^e_u$ for each vertex $u\in V$.
\end{claim}
\begin{proof}
If $u\neq v_i$ for all $0\le i\le k$, then
$$
\sum_{e\in Inc(u)}\widehat{f}^e_u = \sum_{e\in Inc(u)}\mu\ell^e_u=\mu\sum_{e\in Inc(u)}\ell^e_u.
$$
If $u= v_{i+1}$ for some $0\le i< k$, then
\begin{eqnarray*}
\sum_{e\in Inc(u)}\widehat{f}^e_u & = &\widehat{f}^{e_i}_{v_{i+1}} + \widehat{f}^{e_{i+1}}_{v_{i+1}} + \sum_{e\in Inc(v_{i+1})\setminus\{e_{i},e_{i+1}\}}\widehat{f}^e_{v_{i+1}}\\
&=&
\mu(\ell^{e_i}_{v_{i+1}}-\lambda)+
\mu(\ell^{e_{i+1}}_{v_{i+1}}+\lambda) + \sum_{e\in Inc(v_{i+1})\setminus\{e_{i},e_{i+1}\}}\mu\ell^e_{v_{i+1}} \\
&=&
\mu\left(\ell^{e_i}_{v_{i+1}}+\ell^{e_{i+1}}_{v_{i+1}}+\sum_{e\in Inc(v_{i+1})\setminus\{e_{i},e_{i+1}\}}\ell^e_{v_{i+1}}\right)\\
&=&\mu\sum_{e\in Inc(v_{i+1})}\ell^e_{v_{i+1}}.
\end{eqnarray*}
Finally, if $u=v_0$, then, since $\lambda>0$ and $\mu>0$,
\begin{eqnarray*}
\sum_{e\in Inc(u)}\widehat{f}^e_u &=& \widehat{f}^{e_0}_{v_{0}} + \sum_{e\in Inc(v_0)\setminus\{e_{0}\}
}\widehat{f}^e_{v_{0}}\\
&=&
\mu(\ell^{e_0}_{v_{0}}+\lambda)+
\sum_{e\in Inc(v_0)\setminus\{e_{0}\}}\mu\ell^e_{v_0}\\ 
&=&
\mu\left(\ell^{e_0}_{v_0}+\sum_{e\in Inc(v_0)\setminus\{e_0\}}\ell^e_{v_0}\right)+\mu\lambda\\
&=&
\mu\sum_{e\in Inc(v_0)}\ell^e_{v_0}+\mu\lambda\\
&>&
\mu\sum_{e\in Inc(v_0)}\ell^e_{v_0}.
\end{eqnarray*}
The last inequality is true because $\lambda>0$ and $\mu>0$.
\end{proof}

\begin{claim}
 $\widehat r$ is a run of protocol ${\cal P}^\delta_E$ and $\widehat{r}=_h r$.
\end{claim}
\begin{proof}
We need to verify that tuple $\widehat{r}$ satisfies the conditions of Definition~\ref{values} and the local conditions of protocol ${\cal P}^\delta_{E}$ on page~\pageref{local conditions}. Below by $v_{k+1}$  we denote the end of edge $e_k$ different from vertex $v_k$. We start with conditions of Definition~\ref{values}.
\begin{itemize}
\item[{\sf 1(c)}] Due to Claim~\ref{two equalities 3a} and the assumption that $r'$ is a run of protocol ${\cal P}^\delta_{E}$, we only need to verify condition 1(c) for edge $e_k$. 
Note that $\delta_{e_k}\in X^{e_k}$, by Definition~\ref{gamma}. Thus, we only need to show that $\widehat{f}^{e_k}_{v_k}+\widehat{f}^{e_k}_{v_{k+1}}>0$. Indeed, $\ell^{e_k}_{v_k}+\ell^{e_k}_{v_{k+1}}>0$ because run $r'$ satisfies condition 1(c). Since $\lambda>0$ and $\mu > 0$,
$$
\widehat{f}^{e_k}_{v_k}+\widehat{f}^{e_k}_{v_{k+1}}=
\mu(\ell^{e_k}_{v_k}+ \lambda) +\mu\ell^{e_k}_{v_{k+1}} = 
\mu(\ell^{e_k}_{v_k}+\ell^{e_k}_{v_{k+1}}) + \mu\lambda > \mu(\ell^{e_k}_{v_k}+\ell^{e_k}_{v_{k+1}}) > 0.
$$
\item[{\sf 2(a)}] Due to Claim~\ref{two equalities 3a} and the assumption that $r'$ is a run of protocol ${\cal P}^\delta_{E}$, we again only need to verify condition 2(a) for edge $e_k$, which is vacuously true because $\delta_{e_k}\in X^{e_k}$ due to condition 4 of Definition~\ref{gamma}.

\item[{\sf 2(b)}] By the definition of $\widehat{r}$, for each edge $b \in {\cal B} \setminus \{e_0,\dots,e_k\}$, and each vertex $u\in Inc(b)$, we have $\widehat{f}^b_u=\mu\ell^b_u$. Thus, $\widehat{r}$ on any such edge satisfies condition 2(b) of Definition~\ref{values} because run $r'$ does and $\mu>0$. 

We next show that condition 2(b) is satisfied for each $e_i$ such that $e_i\in \cal B$ and $0\le i\le k$. Indeed, consider any $u\in Inc(e_i)$ and suppose that $\widehat{f}^{e_i}_u<0$. 

If $u=v_i$, then, since $\lambda>0$ and $\mu>0$, from equation~(\ref{f hat IIIc}), we have
$$
\ell^{e_i}_u=\ell^{e_i}_{v_i}=\dfrac{\widehat{f}^{e_i}_{v_i}}{\mu}-\lambda<\dfrac{\widehat{f}^{e_i}_{v_i}}{\mu}<0.
$$
Thus, $\Box_{e_i}\bigvee_{e\in C^u_{\miniminus e_i}}\delta_e\in X^{e_i}$ because run $r'$ satisfies condition 2(b) of Definition~\ref{values}.

If $u=v_{i+1}$ and $i<k$, then condition 2(b) is satisfied due to condition 3 of Definition~\ref{gamma}.

Finally, if $i=k$ and $u=v_{k+1}$, then $\widehat{f}^{e_i}_u=\mu\ell^{e_i}_u$ by the definition of $\widehat{r}$. Thus, condition 2(b) is satisfied by run $\widehat{r}$ because it is satisfied by run $r'$ and since $\mu>0$.

\item[{\sf 2(c)}] By the definition of $\widehat{r}$, for each edge $b \in {\cal B} \setminus \{e_0,\dots,e_k\}$, and each vertex $u\in Inc(b)$, we have $\widehat{f}^b_u=\mu\ell^b_u$. Thus, $\widehat{r}$ on any such edge satisfies condition 2(c) of Definition~\ref{values} because run $r'$ does and $\mu>0$. 

We will next show that condition 2(c) is satisfied for each $e_i$ such that $e_i\in \cal B$ and $0\le i< k$. Indeed, note that $\lambda>|\ell^{e_i}_{v_{i+1}}|$ due to the choice of $\lambda$. Thus
$$
\widehat{f}^{e_i}_{v_{i+1}}=\mu(\ell^{e_i}_{v_{i+1}}-\lambda) < 0.
$$
Finally, note that when $i=k$, we have $\delta_{e_k}\in X^{e_k}$. Therefore, condition 2(c) is vacuously true. 

\item[{\sf 3(a)}] Due to Claim~\ref{two equalities 3a} and the assumption that $r'$ is a run of protocol ${\cal P}^\delta_{E}$, we again only need to verify condition 3(a) for edge $e_k$. 
By condition 4 of Definition~\ref{gamma}, $\delta_{e_k}\in X^{e_k}$. Thus, as we have shown in the case 1(c) above, $\widehat{f}^{e_k}_{v_k}+\widehat{f}^{e_k}_{v_{k+1}}>0$. Therefore, condition 3(a) is vacuously true for edge $e_k$.

\item[{\sf 3(b)}] Due to Claim~\ref{two equalities 3a} and the assumption that $r'$ is a run of protocol ${\cal P}^\delta_{E}$, we once more only need to verify condition 3(b) for edge $e_k$. 
By condition 4 of Definition~\ref{gamma}, $\delta_{e_k}\in X^{e_k}$. Thus, condition 3(b) is vacuously true for edge $e_k$.
\end{itemize}
The local conditions (see page~\pageref{local conditions}) are satisfied by tuple $\widehat{r}$ at each vertex $u\in V$ because they are satisfied by run $r'$ and due to Claim~\ref{two equalities 3b} combined with the fact that $\mu>0$.

To show that $\widehat{r}=_h r$, first note that $Y^h=M\cap \Phi(\Sigma,\{h\})=X^h$.  Then, observe that
$$
\widehat{f}^h_{v_0}=\mu(\ell^h_{v_0}+\lambda)=\dfrac{f^h_{v_0}}{\ell^h_{v_0}+\lambda}(\ell^h_{v_0}+\lambda)=f^h_{v_0}.
$$
Finally, note that $f^h_{v_0}=-f^h_{v_1}$ and $\ell^h_{v_0}=-\ell^h_{v_1}$ because runs $r$ and $r'$ satisfy condition 2(a) of Definition~\ref{values}. Thus,
$$
\widehat{f}^h_{v_1}=\mu(\ell^h_{v_1}-\lambda)=\dfrac{f^h_{v_0}}{\ell^h_{v_0}+\lambda}(\ell^h_{v_1}-\lambda)=
\dfrac{-f^h_{v_1}}{-\ell^h_{v_1}+\lambda}(\ell^h_{v_1}-\lambda)=f^h_{v_1}.
$$

\end{proof}
This concludes the proof of Theorem~\ref{run exists 2}.
\end{proof}

\subsection{Aggregated Protocol}\label{aggregation section}

Recall from Section~\ref{canonical protocol section} that canonical protocol $\mathcal{P}^\delta_E$ has formula $\delta$ as a parameter. 
In this section we introduce a construction that aggregates multiple canonical protocols. One can view a run of the aggregated protocol $\mathcal{P}$ as several runs of different canonical protocols for different values of parameter $\delta$ being executed concurrently on different ``levels". Also recall that a value of an edge under a canonical protocol consists of a maximal consistent set of formulas and a pair of real numbers (flow values). Although there is no explicit connection between flow values on different levels for the same edge, we assume that maximal consistent sets are the same on all layers for a given edge of the aggregated protocol, see Definition~\ref{all-values}.

\begin{definition}\label{all-values}
A value $w_e$ of an edge $e\in E$ under the aggregated protocol ${\cal P}$ is a tuple $\langle X, \{f_{v,\delta}\}_{v\in Inc(e),\delta\in\Delta(\Sigma)}\rangle$ such that $\langle X, \{f_{v,\delta}\}_{v\in Inc(e)}\rangle$ is a value of edge $e$ under protocol ${\cal P}^\delta_E$ for each $\delta\in\Delta(\Sigma)$.
\end{definition}

\paragraph{Valuation.} Let $\pi$ be a function such that, for each $e\in E$ and $p\in P_e$, set $p^\pi$ contains all values $\langle X, \{f_{v,\delta}\}_{v\in Inc(e),\delta\in\Delta(\Sigma)}\rangle$,  where $p\in X$.   

\paragraph{Local Conditions.}  A tuple 
$\langle X^e, \{f^e_{v,\delta}\}_{v\in Inc(e),\delta\in\Delta(\Sigma)}\rangle_{e\in Inc(u)}$
satisfies the local conditions of protocol $\cal P$ at vertex $u$ if for each $\delta\in\Delta(\Sigma)$, the tuple
$\langle X^e, \{f^e_{v,\delta}\}_{v\in Inc(e)}\rangle_{e\in Inc(u)}$
satisfies local conditions of protocol ${\cal P}^\delta_E$ at vertex $u$.

This concludes the definition of the aggregated protocol ${\cal P}$.

\begin{theorem}\label{main induction}
If $e\in E$, $\phi\in \Phi(Sig,\{e\})$, and tuple $$r=\langle X^h, \{f^h_{u,\delta} \}_{u\in Inc(h),\delta\in\Delta(\Sigma)}\rangle_{h\in E}$$ is a run of protocol ${\cal P}$, then $r\Vdash\phi$ if and only if $\phi\in X^e$.
\end{theorem}
\begin{proof}
We prove the theorem by induction on the structural complexity of formula $\phi$. If $\phi$ is a proposition $p\in P_e$, then the required follows from Definition~\ref{sat} and the definition of valuation function $\pi$ for protocol $\cal P$. The cases when $\phi$ is constant $\bot$ or an implication $\phi_1\to\phi_2$ follow from Definition~\ref{sat} and the maximality and the consistency of set $X^e$ in the standard way. Now let $\phi$ be of the form $\Box_e\psi$.

\noindent $(\Rightarrow):$ 
Suppose that $\bigwedge_{i}\bigvee_{h\in E}\psi^i_h$ is the conjunctive normal form of $\neg\psi$ such that $\psi^i_h\in\Phi(\Sigma,\{h\})$ for each $h\in E$. Thus, the following statement can be proven using just the axioms of the propositional logic in language $\Phi(\Sigma)$
\begin{equation}\label{wedge imp psi}
\vdash \neg\bigwedge_{i}\bigvee_{h\in E}\psi^i_h \to \psi.
\end{equation}

Assume that $\Box_e\psi\notin X^e$. To prove that $r\nVdash\Box_e\psi$, it suffices to show that there is a run $\widehat{r}$ of the canonical protocol ${\cal P}_E$ such that $\widehat{r}=_e r$ and $\widehat{r}\Vdash \bigwedge_{i}\bigvee_{h\in E}\psi^i_h$. 

The assumption $\Box_e\psi\notin X^e$ and the maximality of set $X^e$ imply that $X^e\nvdash \Box_e\psi$. Thus, $X^e\nvdash \psi$ by Lemma~\ref{XYZ}. Hence,
set $X^e\cup \{\neg\psi\}$ is consistent. Let $M$ be any maximal consistent extension of $X^e\cup \{\neg\psi\}$. By Theorem~\ref{run exists 2}, for each $\delta\in\Delta(\Sigma)$ there is a run $\widehat{r}_\delta=\langle \widehat{X}^h, \{\widehat{f}^h_{u,\delta} \}_{u\in Inc(h)}\rangle_{h\in E}$ of the canonical protocol ${\cal P}^\delta_E$ such that $\widehat{r}=_e r$ and $\widehat{X}^h=M\cap \Phi(\Sigma,\{h\})$ for each $h\in E$. Define tuple $\widehat{r}$ to be $\langle \widehat{X}^h, \{\widehat{f}^h_{u,\delta} \}_{u\in Inc(h),\delta\in\Delta(\Sigma)}\rangle_{h\in E}$. By the definition of protocol $\cal P$, tuple $\widehat{r}$ is a run of $\cal P$. 

We next show that $\widehat{r}\Vdash \bigwedge_{i}\bigvee_{h\in E}\psi^i_h$. Suppose the opposite, then there is $i_0$ such that $\widehat{r}\nVdash \bigvee_{h\in E}\psi^{i_0}_h$. Thus, $\widehat{r}\nVdash\psi^{i_0}_h$ for each $h\in E$. Hence, by the induction hypothesis, $\psi^{i_0}_h\notin \widehat{X}^h$ for each $h\in E$. Recall that $\psi^{i_0}_h\in\Phi(\Sigma,\{h\})$ and $\widehat{X}^h$ is a maximal consistent subset of $\Phi(\Sigma,\{h\})$ for each $h\in E$. Thus, $\neg\psi^{i_0}_h\in \widehat{X}^h\subseteq M$ for each $h\in E$. Hence, $\bigwedge_{h\in E}\neg\psi^{i_0}_h\in M$ due to maximality of the set $M$. Then,
$M\vdash \neg \bigvee_{h\in E}\psi^{i_0}_h$. Hence, $M\vdash\neg \bigwedge_{i}\bigvee_{h\in E}\psi^i_h$. Therefore, $M\vdash\psi$, by statement (\ref{wedge imp psi}). The latter contradicts the choice of set $M$ being a maximal consistent extension of set $X^e\cup \{\neg\psi\}$.

\noindent $(\Leftarrow):$ Suppose that $\Box_e\psi\in X^e$. We will show that $r\Vdash\Box_e\psi$. Consider any run $\widehat{r}=\langle \widehat{X}^h, \{\widehat{f}^h_{u,\delta} \}_{u\in Inc(h),\delta\in\Delta(\Sigma)}\rangle_{h\in E}$ of the aggregated protocol ${\cal P}$ such that $\widehat{r}=_e r$. It suffices to prove that $\widehat{r}\Vdash\psi$.

Let  $\bigwedge_{i}\bigvee_{h\in E}\psi^i_h$ be a conjunctive normal form of $\psi$ such that $\psi^i_h\in\Phi(\Sigma,\{h\})$ for each $h\in E$. Then, for each $i$, the following statement can be proven using just the axioms of the propositional logic in language $\Phi(\Sigma)$
\begin{equation*}
\vdash \psi\to\bigvee_{h\in E}\psi^i_h.
\end{equation*}
By Necessitation inference rule 
\begin{equation*}
\vdash \Box_e\left(\psi\to\bigvee_{h\in E}\psi^i_h\right).
\end{equation*}
By Distributivity axiom and Modus Ponens inference rule,
\begin{equation*}
\vdash \Box_e\psi\to\Box_e\bigvee_{h\in E}\psi^i_h.
\end{equation*}
Thus, for each $i$, we have $\Box_e\bigvee_{h\in E}\psi^i_h\in X^e$ due to the assumption $\Box_e\psi\in X^e$ and the maximality of set $X^e$. Note that $\widehat{X}^e=X^e$ due to the assumption $\widehat{r}=_e r$. Hence, $\Box_e\bigvee_{h\in E}\psi^i_h\in \widehat{X}^e$. Let $\widehat{\delta}$ denote the formula $\bigvee_{h\in E}\psi^i_h$. Recall that $\widehat{r}$ is a run of protocol $\cal P$. Hence, by the definition of the aggregated protocol, tuple $\langle \widehat{X}^h, \{\widehat{f}^h_{u,\widehat{\delta}} \}_{u\in Inc(h)}\rangle_{h\in E}$ is a run of protocol ${\cal P}^{\widehat{\delta}}_{E}$, and so, by Lemma~\ref{sub protocol}, it is a run of protocol ${\cal P}^{\widehat{\delta}}_{\{e\}}$.
Then, by Theorem~\ref{XYZ2}, there is an edge $h_0\in E$ such that $\psi^i_{h_0}\in \widehat{X}^{h_0}$. Thus, by the induction hypothesis, $\widehat{r}\Vdash\psi^i_{h_0}$. Hence, $\widehat{r}\Vdash \bigvee_{h\in E}\psi^i_h$ for each $i$. Then, $\widehat{r}\Vdash \bigwedge_{i}\bigvee_{h\in E}\psi^i_h$. Therefore, $\widehat{r}\Vdash\psi$.
\end{proof}

\begin{theorem}[completeness]
For any signature $\Sigma$ and any formula $\phi\in\Phi(\Sigma)$, if $\nvdash\phi$, then there exists a protocol ${\cal P}$ over $\Sigma$ and a run $r$ of ${\cal P}$ such that $r\nVdash\phi$. 
\end{theorem}
\begin{proof}
Suppose that $\nvdash\phi$. Let $M$ be a maximal consistent subset of $\Phi(\Sigma)$ containing the formula $\neg\phi$. Assume that $\bigwedge_i\bigvee_{e\in E}\phi^i_e$ is the conjunctive normal form of the formula $\neg\phi$ such that $\phi^i_e\in\Phi(\Sigma, \{e\})$ for each $i$ and each $e\in E$. Since $\neg\phi\in M$, 
for each $i$ there exists $e_i\in E$ such that $\phi^i_{e_i}\in M$.
By Theorem~\ref{run exists}, for each $\delta\in\Delta(\Sigma)$, there exists a run $r^\delta=\langle X^h, \{f^h_{u} \}_{u\in Inc(h)}\rangle_{h\in E}$ of the canonical protocol ${\cal P}^\delta_E$ such that $X^h= M\cap \Phi(\Sigma, \{h\})$ for all $h\in E$. Thus, $\phi^i_{e_i}\in X^{e_i}$ for each $i$. Consider tuple $r=\langle X^h, \{f^h_{u,\delta} \}_{u\in Inc(h),\delta\in\Delta(\Sigma)}\rangle_{h\in E}$. By the definition of the aggregated protocol, tuple $r$ is a run of protocol $\cal P$.
 Hence, $r\vDash\phi^i_{e_i}$ for each $i$, by Theorem~\ref{main induction}. Therefore, $r\Vdash \bigwedge_i\bigvee_{e\in E}\phi^i_e$ and so $r\Vdash\neg\phi$. 
\end{proof}

\section{Conclusion}\label{conclusion section}

In this article we have developed a formal modal logical framework for reasoning about information flow in communication networks with a fixed topological structure. Our main results are the soundness and the completeness of this logical system. At the core of the proof of the completeness is a well-known network flow protocol. A natural possible extension of this work is to develop a similar system for directed graphs that represent networks with one-way communication channels. Another possible extension is a distributed knowledge system with a modality $\Box_{A}$ in which the statement $\Box_A\phi$ is interpreted as ``any agent that eavesdrops on all channels in set $A$ knows that $\phi$ is true".

Another possible direction for the future work is to develop logical frameworks for reasoning about information flow in more specialized settings. An example of such a setting is the influence flow in social networks. The influence in social networks is usually modeled by a relatively simple and very specific form of ``local conditions" such as those in commonly used threshold model~\cite{v96sn,macy91asr,kkt03sigkdd,am14fi,g78ajs,s78}. A logical framework for such a setting is likely to include more powerful version of Gateway axiom. The canonical network construction for the proof of the completeness presented in this article is very unlikely to be adoptable to a much more restricted interpretation of local conditions found in social network.

\bibliographystyle{unsrt}
\bibliography{sp}

\end{document}